\documentclass{article}

\usepackage{latexsym,amsmath,amssymb,url}
\usepackage{amsthm}
\usepackage{todonotes}
\usepackage{graphicx}
\usepackage{subcaption}
\usepackage[ruled,vlined]{algorithm2e}
\usepackage{multirow}

\newcommand{\taxa}{X}
\newcommand{\cha}{\chi}
\newcommand{\ext}{\phi}
\newcommand{\extScore}{\Delta}
\newcommand{\PS}[2]{l_{#1}(#2)}
\newcommand{\states}{\bf{C}}

\newcommand{\dmp}{d_{MP}}
\newcommand{\damp}{d_{AMP}}
\newcommand{\problemName}{{\sc dmp}}
\newcommand{\dtbr}{d_{TBR}}

\newtheorem{lemma}{Lemma}
\newtheorem{corollary}{Corollary}
\newtheorem{theorem}{Theorem}
\newtheorem{observation}{Observation}
\newtheorem{krule}{Reduction Rule}

\newcommand{\steven}[1]{\textcolor{blue}{#1}}

\title{Maximum parsimony distance on phylogenetic trees: a linear kernel and constant factor approximation algorithm}

\author{Mark Jones\footnote{Centrum Wiskunde \& Informatica (CWI), 1098 XG Amsterdam, The Netherlands} \footnote{Corresponding author, contact: markelliotlloyd@gmail.com}
\and Steven Kelk\footnote{Department of Data Science and Knowledge Engineering (DKE),
Maastricht University,
6200 MD Maastricht,
The Netherlands} 
\and Leen Stougie\footnotemark[1] \footnote{Vrije Universiteit Amsterdam, 1081 HV Amsterdam,
The Netherlands} \footnote{INRIA-Erable}}

\begin{document}
\maketitle

\abstract{Maximum parsimony distance is 
a measure used to quantify the dissimilarity of two unrooted phylogenetic trees.  It is NP-hard to compute, and very few positive algorithmic results are known due to its complex combinatorial structure. Here we address this shortcoming by showing that the problem is fixed parameter tractable. We do this by establishing a linear kernel i.e., that after applying certain reduction rules the resulting instance has size that is bounded by a linear function of the distance. As powerful corollaries to this result we prove that the problem permits a polynomial-time constant-factor approximation algorithm; that the treewidth of a natural auxiliary graph structure encountered in phylogenetics is bounded by a function of the distance; and that the distance is within a constant factor of the size of a maximum agreement forest of the two trees, a well studied object in phylogenetics.

}

\section{Introduction}

Phylogenetics is the science of inferring and comparing trees (or more generally, graphs) that represent the evolutionary history of a set of species \cite{SempleSteel2003}. In this article we focus on trees. The inference problem has been comprehensively studied: given only data about the species in $X$ (such as DNA data) construct a \emph{phylogenetic tree} which optimizes a particular objective function \cite{felsenstein2004inferring,warnow2017computational}. Informally, a phylogenetic tree is simply a tree whose leaves are bijectively labelled by $X$. Due to different objective functions, multiple optima and the phenomenon that certain genomes are the result of several evolutionary paths (rather than just one) we are often confronted with multiple ``good'' phylogenetic trees \cite{nakhleh2013computational}. In such cases we wish to formally quantify how dissimilar these trees really are. This leads naturally to the problem of defining and computing the \emph{distance} between phylogenetic trees \cite{katherine2017shape}. Many such distances have been proposed, some of which can be computed in polynomial-time, such as \emph{Robinson-Foulds} (RF) distance \cite{robinson1981comparison}, and some of which are NP-hard, such as \emph{Subtree Prune and Regraft} (SPR) distance \cite{bonet2010complexity} or \emph{Tree Bisection and Reconnection} (TBR) distance \cite{AllenSteel2001}.

Interestingly, distances are not only relevant as a numerical quantification of difference: they also appear in constructive methods for the inference of phylogenetic networks \cite{HusonRuppScornavacca10}, which generalise trees to graphs, and phylogenetic supertrees, which seek to merge multiple trees into a single summary tree \cite{Whidden2014}. In recent decades NP-hard phylogenetic distances have attracted quite some attention from the discrete optimization and parameterized complexity communities, see e.g. \cite{bulteau2019parameterized,downey2013fundamentals}.


In this article we focus on a relatively new distance measure, \emph{maximum parsimony distance}, henceforth denoted $d_{MP}$. Let $T_1$ and $T_2$ be two unrooted (i.e. undirected) binary phylogenetic trees, with the same set of leaf labels $X$. Consider an arbitrary assignment of colours (``states'') to $X$; we call such an assignment a \emph{character}. The \emph{parsimony score} of $T_1$ with respect to the character is the minimum number of bichromatic edges in $T_1$, ranging over all possible colourings of the internal vertices of $T_1$. The parsimony distance of $T_1$ and $T_2$ is the maximum absolute difference between parsimony scores of $T_1$ and $T_2$, ranging over all characters \cite{fischerNonbinary,dMP-moulton}. 

The distance has several attractive properties; it is a metric, and (unlike e.g. RF distance) it is not confounded by the influence of horizontal evolutionary events \cite{fischerNonbinary}. Furthermore, the concept of parsimony, which lies at the heart of $d_{MP}$, is fundamental in phylogenetics since it articulates the idea that explanations of evolutionary history should be no more complex than necessary. Alongside its historical significance for applied phylogenetics \cite{felsenstein2004inferring}, the study of character-based parsimony has given rise to many beautiful combinatorial and algorithmic results; we refer to e.g. \cite{steel2016phylogeny, liers2016binary,strike,AlonApproxMP,moran2008convex} for overviews.

Unfortunately, it is NP-hard to compute $d_{MP}$ \cite{kelk2014complexity}. A simple exponential-time algorithm is known \cite{kelk2017note}, which runs in time $O( \phi^n \cdot \text{poly}(n))$, where $|X|=n$ and $\phi \approx 1.618$ is the golden ratio, but beyond this few positive results are known. This is frustrating and surprising, since a number of results link $d_{MP}$ to the well-studied TBR distance, henceforth denoted $d_{TBR}$. 
Namely, it has been proven that $d_{MP}$ is a lower bound on $d_{TBR}$ \cite{fischerNonbinary}, which, informally, asks for the minimum number of topological rearrangement operations to transform one tree into the other; an empirical study has suggested that in practice the distances are often very close \cite{kelk2016reduction}. Also, $d_{MP}$ has been used to prove the tightness of the best-known kernelization results for $d_{TBR}$ \cite{tightkernel,kelk2019new}. What, exactly, is the relationship between $d_{MP}$ and $d_{TBR}$? This is a pertinent question, which transcends the specifics of TBR distance because, crucially, $d_{TBR}$ can be characterized using the powerful \emph{maximum agreement forest} abstraction.

Distances based on agreement forests have been intensively and successfully studied in recent years, as the use of the agreement forest abstraction almost always yields fixed parameter tractability and constant-factor approximation algorithms \cite{bordewich2017fixed}, many of which are effective in practice. We refer to \cite{whidden2013fixed,van2016hybridization,chen2016approximating,shi2018parameterized} for recent overviews of the agreement forest literature, and books such as \cite{cyganBook} for an introduction to fixed parameter tractability. In particular, $d_{TBR}$ can be computed in $O(3^{d_{TBR}} \cdot \text{poly}(n))$ time \cite{chen2015parameterized}, permits a polynomial-time 3-approximation algorithm, and a kernel of size $11d_{TBR} - 9$ \cite{kelk2019new}. 

In contrast, prior to this paper very little was known about $d_{MP}$:  nothing was known about the approximability of $d_{MP}$; it was not known whether it is fixed parameter tractable (where $d_{MP}$ is the parameter); and, while, as mentioned above, it is known that $d_{MP} \leq d_{TBR}$, it remained unclear how much smaller $d_{MP}$ can be than $d_{TBR}$ in the worst case. Despite promising partial results it even remained unclear whether questions such as ``Is $d_{MP} \geq k$?'' can be solved in \emph{polynomial} time when $k$ is a constant \cite{boes2016linear,kelk2016reduction}. This is another important difference with distances such as $d_{TBR}$, where corresponding questions are trivially polynomial time solvable for fixed $k$. The apparent extra complexity of $d_{MP}$ seems to stem from the unusual max-min definition of the problem, and the fact that unlike $d_{TBR}$, which is based on topological rearrangements of subtrees, $d_{MP}$ is based only on characters.

In this article we take a significant step forward in understanding the deeper complexity of $d_{MP}$ and resolve all of the above questions. Our central result is that we prove that two common polynomial-time reduction rules encountered in phylogenetics, the \emph{subtree} and \emph{chain} reductions \cite{AllenSteel2001}, are sufficient to produce a \emph{linear kernel} for $d_{MP}$. This means that, after exhaustive application of these rules, which preserve $d_{MP}$, the reduced trees will have at most $\alpha \cdot d_{MP}$ leaves, with $\alpha =560$. 
The fixed parameter tractability of computing $d_{MP}$ (parameterized by itself) then follows, by solving the kernel using the exact algorithm from \cite{kelk2017note}. The fact that the reduction rules preserve $d_{MP}$ was already known \cite{kelk2016reduction}. However, proving the bound on the size of the reduced trees requires rather involved combinatorial arguments, which have a very different flavour to the arguments typically encountered in the maximum agreement forest literature. 
The main goal of this article is to present these arguments as clearly as possible, rather than to optimize the resulting constants.

The kernel confirms that questions such as ``Is $d_{MP} \geq k$?'' can, indeed, be solved in polynomial time: it is striking that here the proof of fixed parameter tractability has preceded the weaker result of polynomial-time solveability for fixed $k$. 

Next, by producing a 
modified, constructive version
of the bounding argument underpinning the kernelization,
we are able to demonstrate a polynomial-time $\alpha(1+1/r)$-factor approximation algorithm for computation of $d_{MP}$ for any constant $r$, placing the problem in APX. 

A number of other powerful corollaries result from the kernelization. We leverage the fact that the reduction rules also preserve $d_{TBR}$, to show that $1 \leq \frac{d_{TBR}}{d_{MP}} \leq 2\alpha$, which limits how much smaller $d_{MP}$ can be than $d_{TBR}$. 
Subsequently, we show that the treewidth of an auxiliary graph structure known as the \emph{display graph} \cite{bryant2006compatibility} is bounded by a linear function of $d_{MP}$, resolving an open question posed several times \cite{kelk2015,kelk2016reduction}. The treewidth bound, and the existence of a non-trivial approximation algorithm for $d_{MP}$, were specified as sufficient conditions for proving the fixed parameter tractability of $d_{MP}$ via Courcelle's Theorem \cite{kelk2016reduction}; our linear kernel implies \emph{them}. Summarising, our central result shows how kernelization can open the gateway to a host of strong auxiliary results and bypass intermediate steps in the algorithm design process.

The structure of the paper is as follows. In Section~\ref{sec:defs}  we give formal definitions and insightful preliminary results. In Section~\ref{sec:kernel}
we prove our main result: the linear kernel. The section starts with Subsection~\ref{sec:kerneloverview} that gives a high-level overview of how a sequence of lemmas and theorems lead to the kernel, whereas in the rest of the section these lemmas and theorems are proved. Interesting corollaries of the existence of a linear kernel are derived in Section~\ref{sec:cor}: A constant approximation algorithm in Section~\ref{sec:approx}; A bound on the ratio between $d_{MP}$ and $d_{TBR}$ in Section~\ref{sec:dtbr}; A bound on the treewidth of the so-called display graph in terms of $d_{MP}$ in Section~\ref{sec:tw}. Section~\ref{sec:concl} concludes with some directions for future research.

\section{Definitions and Preliminaries}
\label{sec:defs}

\begin{figure}
    \begin{subfigure}[b]{0.45\textwidth}
    \includegraphics[width=\textwidth]{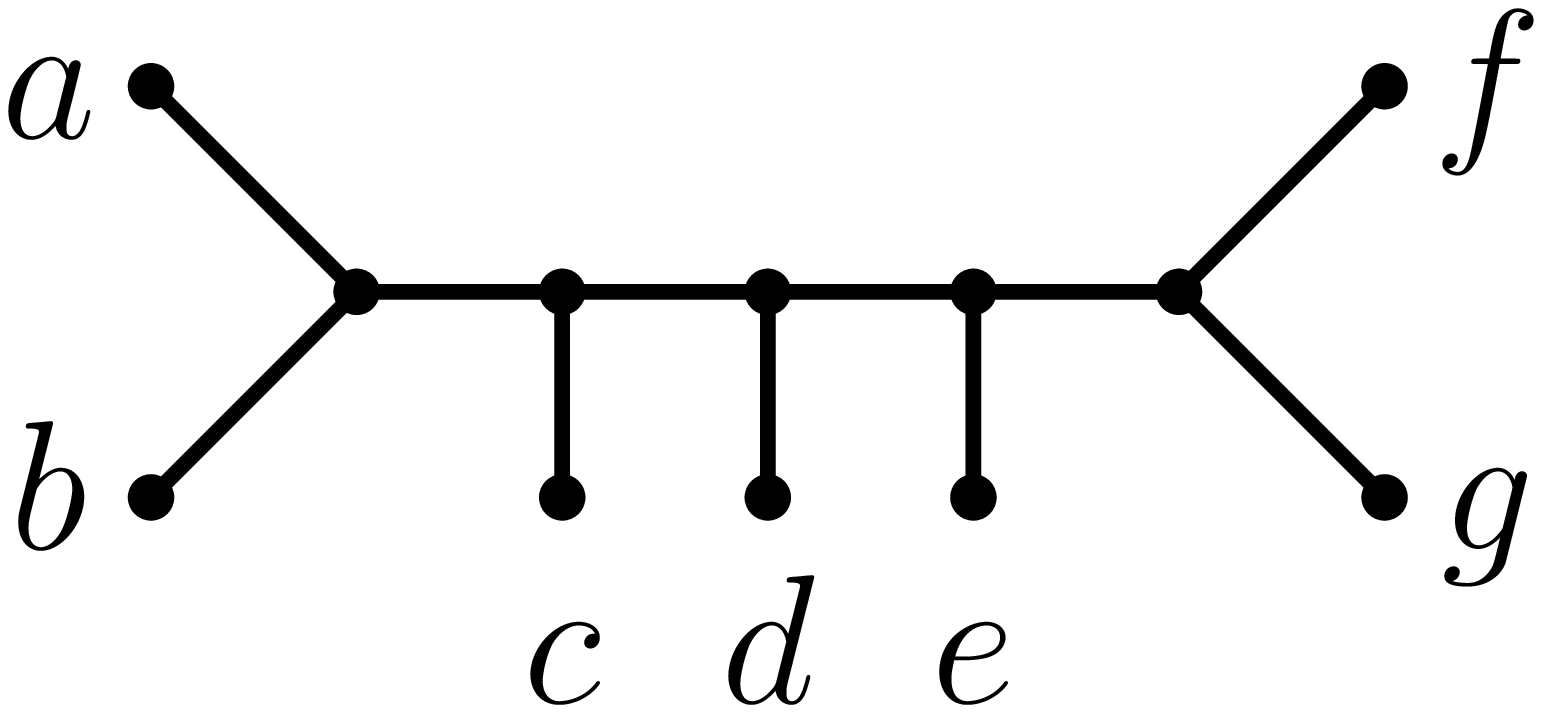}
    \caption{$T_1$}
    \end{subfigure}
    ~
    \begin{subfigure}[b]{0.45\textwidth}
    \includegraphics[width=\textwidth]{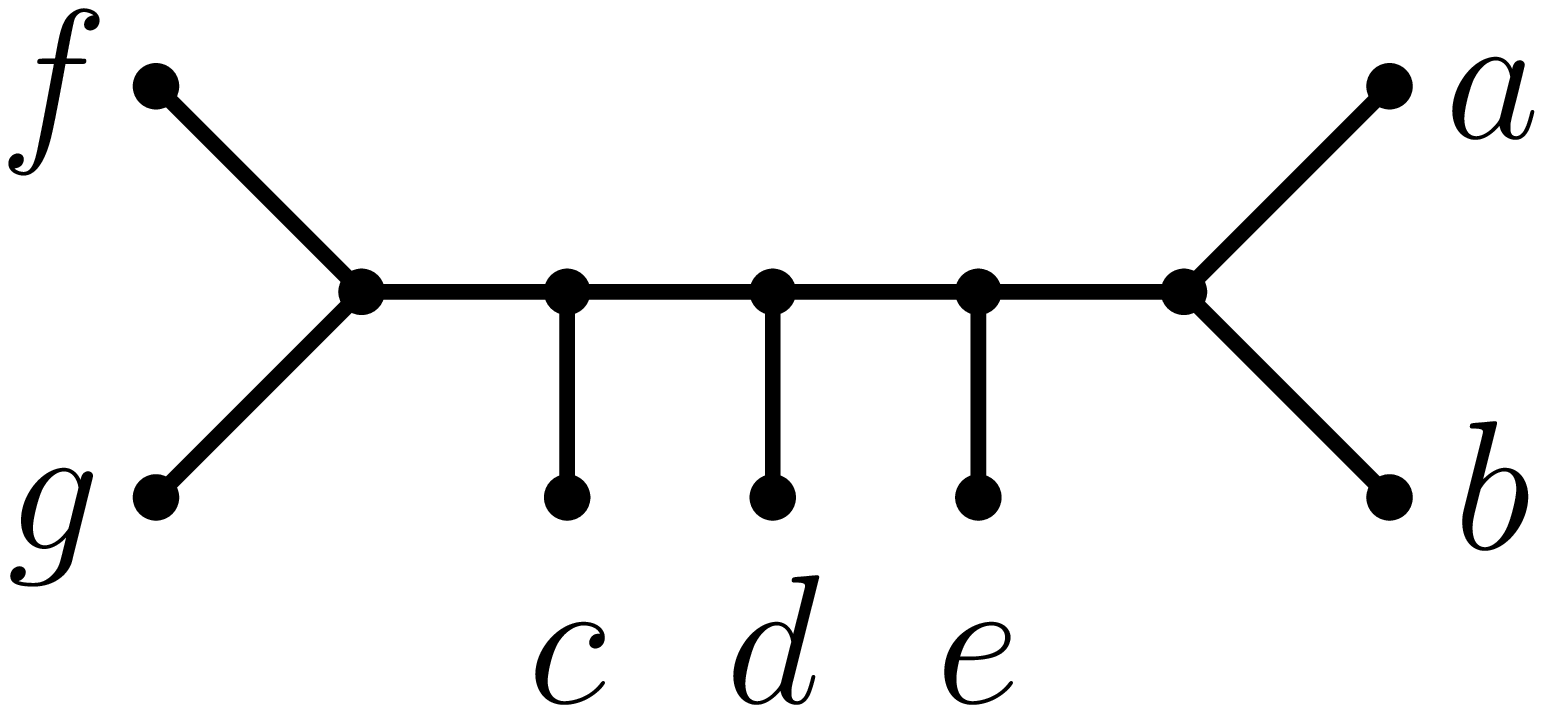}
    \caption{$T_2$}
    \end{subfigure}
    \caption{Two unrooted binary phylogenetic trees $T_1, T_2$ on $X=\{a,\ldots,g\}$. A character which assigns colour \textsc{red} to $\{a,b,c\}$ and \textsc{blue} to
    $\{d,e,f,g\}$ has parsimony score 1 on the left tree, and 2 on the right, proving that $d_{MP}(T_1, T_2) \geq |1 - 2| = 1$. In fact, it can be verified that no character can cause the parsimony scores of these two trees to differ by more, so $d_{MP}(T_1, T_2)=1$. As noted in Section \ref{sec:dtbr}, $d_{TBR}(T_1, T_2)=2$, because a maximum agreement forest of these two trees contains three blocks \cite{kelk2016reduction}.
    } 
    \label{fig:exampleMP}
\end{figure}


An \emph{unrooted binary phylogenetic tree} on a set of species (or taxa) $\taxa$ is an undirected tree in which all internal vertices have degree $3$, and the degree-$1$ vertices (the \emph{leaves}) are bijectively labelled with elements from $\taxa$. For brevity we will refer to unrooted binary phylogenetic trees as \emph{phylogenetic trees}, or even shorter \emph{trees}. See Figure~\ref{fig:exampleMP} for an example.


Given a set $S \subseteq \taxa$ and a tree $T$ on $\taxa$, we denote by $T[S]$ the \emph{spanning subtree on $S$ in $T$}, that is, the minimal connected subgraph $T'$ of $T$ such that $T'$ contains every element of $S$.
The \emph{induced subtree $T|_S$ by $S$ in $T$} is the tree derived from $T[S]$ by suppressing any vertices of degree $2$. 

Given a subset $S \subseteq \taxa$ and a tree $T$ on $\taxa$, we say that $S$ has \emph{degree $d$ in $T$} 
if there are exactly $d$ edges $uv$ in $T$ for which $u$ is in $T[S]$ and $v$ is not; in other words, $d$ is the number of edges separating $T[S]$ from the rest of $T$. We call these edges \emph{pending edges}.
of $T$. 

For two disjoint subsets $S_1,S_2 \subseteq \taxa$, we say $S_1$ and $S_2$ are \emph{spanning-disjoint} in $T$ if the spanning subtrees $T[S_1]$ and $T[S_2]$ are edge-disjoint. (Observe that as $T$ is binary, this also implies that $T[S_1]$ and $T[S_2]$ are vertex-disjoint.)
Similarly, we say a collection $S_1, \dots S_m$ of subsets of $\taxa$ are \emph{spanning-disjoint} in $T$ if $S_i, S_j$ are spanning-disjoint in $T$ for any $i\neq j$.

\subsection{Characters and parsimony}

A \emph{character} on $\taxa$ is a function $\cha:\taxa \rightarrow \states$, where $\states$ is a set of \emph{states}. In this paper there is no limit on the size of $\states$, in contrast to some contexts where $|\states|$ is assumed to be quite small (for example, in genetic data the nucleobases A,C,G,T). 
Think of the states as colours, say $1,2,\ldots,t=:[t]$.

For a given character $\cha$ and tree $T$ on $\taxa$, the \emph{parsimony score} measures how well $T$ fits $\cha$. It is defined in the following way.
Call a colouring $\ext^{\cha}:V(T) \rightarrow [t]$ an \emph{extension} of $\cha$ to $T$ if   $\ext^{\cha}(x) = \cha(x)$ for all $x \in \taxa$. We usually omit superscript $\cha$ of $\ext$ if the character is clear from the context. Denote by $\extScore_T(\ext)$ the number of bichromatic edges $uv$ in $T$, i.e. for which $\ext(u) \neq \ext(v)$.
Again, we usually omit subscript $T$ when the tree is clear from context.
The \emph{parsimony score} for $T$ with respect to $\cha$ is defined as 
$$\PS{\cha}{T} = \min_{\ext}\extScore_T(\ext)$$
where the minimum is taken over all possible extensions $\ext$ of $\cha$ to $T$. 
An extension $\ext$ that achieves this bound is called an \emph{optimal extension} of $\cha$ to $T$. An optimal extension, and thus the parsimony score, can be easily computed in polynomial time using dynamic programming or e.g. Fitch's algorithm \cite{fitch1971}.

Observe that for any $T$ and $\cha$, the parsimony score for $T$ with respect to $\cha$ is at least $|\cha(\taxa)| - 1$, i.e. the number of colours assigned by $\cha$ minus $1$. If $\PS{\cha}{T}$ is exactly $|\cha(\taxa)| - 1$, we say that $T$ is a \emph{perfect phylogeny} for $\cha$.
For trees $T_1,T_2$ and a character $\cha$ on $X$, the \emph{parsimony distance with respect to $\cha$} is defined as 
$${\dmp}_{\cha}(T_1,T_2) = |\PS{\cha}{T_1} - \PS{\cha}{T_2}|.$$ 

Now we are ready to define the \emph{maximum parsimony distance} between two trees (see also Figure~\ref{fig:exampleMP}). For two trees $T_1, T_2$ on $X$, the maximum parsimony distance is defined as
$$\dmp(T_1,T_2) = \max_{\cha}{\dmp}_{\cha}(T_1,T_2)$$
where the maximum is taken over all possible characters $\cha$ on $X$ \cite{fischerNonbinary,dMP-moulton}. Equivalently, we may write it as 

$$\dmp(T_1,T_2) = \max_{\cha}|\extScore(\ext^{\cha}_1) - \extScore(\ext^{\cha}_2)|$$
where $\ext^{\cha}_1$ is an optimal extension of $\cha$ to $T_1$, and $\ext^{\cha}_2$ an optimal extension of $\cha$ to $T_2$.
This measure satisfies the properties of a distance metric on the space of unrooted binary phylogenetic trees \cite{fischerNonbinary,dMP-moulton}. For two 
trees on $n$ taxa 
it is known that $d_{MP}$ is at most
$n - 2\sqrt{n}+1$ \cite{fischerNonbinary}. A weaker bound of $n-1$ is easily obtained by observing that the parsimony score of a character on a tree is at least 0 and at most $n-1$.





Given a tree $T$ on $\taxa$ and a colouring $\ext:V(T) \rightarrow [t]$,
the \emph{forest induced by $\ext$} is derived from $T$ by deleting every bichromatic edge under $\ext$.
Observe that the number of connected components in the forest induced by $\ext$ is exactly $\extScore(\ext)+1$.
\begin{lemma}\label{lem:PSlowerBound}
If $\cha:\taxa \rightarrow [t]$ is a character with $S_i  = \cha^{-1}(i) \neq \emptyset $ for all $i \in [t]$ and $T$ is a tree on $\taxa$, then
$$\PS{T}{\cha} \geq t-1$$
with equality if and only if $S_1, \dots S_t$ are spanning-disjoint in $T$.
\end{lemma}
\begin{proof}
To see that $\PS{T}{\cha} \geq t-1$, consider an optimal extension $\ext$ of $\cha$ to $T$, and let $F$ be the forest induced by $\ext$.
As each connected component in $F$ is monochromatically coloured by $\ext$, there must be at least $t$ connected components, and thus $\extScore(\ext) \geq t-1$, which implies $\PS{\cha}{T} \geq t-1$.

Now suppose that $S_1, \dots, S_t$ are spanning-disjoint in $T$. Then construct an extension $\ext$ of $\cha$ to $T$ by first setting $\ext(u) = i$ for every vertex $u$ in $T[S_i]$ , for each $i \in [t]$. (As the spanning trees are edge-disjoint and thus vertex-disjoint in $T$, this is well-defined). For any remaining unassigned vertices $v$, if $v$ has a neighbour $u$ for which $\ext(u)$ is defined, then set $\ext(v) = \ext(u)$. Repeat this process until every vertex is assigned a colour by $\ext$.
Now observe that by construction, the vertices assigned colour $i$ by $\ext$ form a connected subtree for each $i \in [t]$.
Thus the forest induced by $\ext$ has exactly $t$ connected components, and so $\extScore(\ext) = t-1$.

Finally, suppose $\PS{\cha}{T} = t-1$, and let $\ext$ be an optimal extension of $\cha$. Then the forest $F$ induced by $\ext$ has exactly $t$ connected components, which implies by the pigeonhole principle that each $S_i$ is a subset of one connected component in $F$. Then as each $S_i$ is contained within a different connected component of $F$, the spanning trees $T[S_i]$ are also contained within these components, and so $S_1, \dots S_t$ are spanning-disjoint.
\end{proof}

\subsection{Parameterized complexity and kernelization}

A \emph{parameterized problem} is a problem for which the inputs are of the form $(x,k)$, where $k$ is an non-negative integer, called the \emph{parameter}. A parameterized problem is \emph{fixed-parameter tractable} (FPT) if there exists an algorithm that solves any instance $(x,k)$ in $f(k)\cdot |x|^{O(1)}$ time, where $f()$ is a computable function depending only on $k$.
A parameterized problem has a \emph{kernel} of size $g(k)$ if there exists a polynomial time algorithm transforming any instance $(x,k)$ into an equivalent problem $(x',k')$, with $|x|,k' \leq g(k)$.
If $g(k)$ is a polynomial in $k$ then we call this a \emph{polynomial kernel}; if $g(k) = O(k)$ then it is a \emph{linear kernel}. It is well-known that that a parameterized problem is fixed-parameter tractable if and only if it has a (not necessarily polynomial) kernel. For more information, we refer the reader to~\cite{ParamAlg2015}.

For a maximization problem $\Pi$ and $\rho \geq 1$, we say $\Pi$ has a \emph{constant factor approximation} with \emph{approximation ratio $\rho$} if there exists a polynomial-time algorithm such that for any instance $\pi$ of $\Pi$, the following inequalities hold, where $opt(\pi)$ denotes the maximum value of a solution to $\pi$, and $alg(\pi)$ denotes the value of the solution to $\pi$ returned by the algorithm: 
$$1 \leq \frac{opt(\pi)}{alg(\pi)} \leq \rho$$

In this paper we study the following maximization problem:

\medskip 
\fbox{
\parbox{0.75\textwidth}{
{\problemName} \\
{\bf Input:} Two trees $T_1,T_2$ on a set of taxa $\taxa$. \\
{\bf Output:} A character $\cha$ on $\taxa$ that maximizes $|\PS{\cha}{T_1} - \PS{\cha}{T_2}|.$ \\
}
}

\medskip

\section{Kernel bound}
\label{sec:kernel} 

\subsection{Overview}
\label{sec:kerneloverview}

In this section we give an overview of the constituent parts of our 
kernelization result, and how they fit together.

The first step is to apply two reduction rules, described in the next section. Rules 1 and 2 correspond roughly to the Cherry and Chain reduction rules that often appear in papers on computational phylogenetics. The correctness of these rules was proved in~\cite{kelk2016reduction}; our contribution is to show that the exhaustive application of these rules grants a linear kernel, as stated in the following theorem.

\newcommand{\kernelBoundText}{
There exists a constant $\alpha$ ($\alpha = 560$) for which the following holds.
Let $(T_1,T_2)$ be a pair of binary unrooted phylogenetic trees on $\taxa$ that are irreducible under Reduction Rules~\ref{rule:cherry} and~\ref{rule:chain}.

Then if $|\taxa| \geq \alpha k$, it holds that $\dmp(T_1,T_2) \geq k$, and we can find a witnessing character, i.e. a character $\cha$ yielding ${\dmp}_{\cha}(T_1,T_2) \geq k$, 
in polynomial time.}
\begin{theorem}\label{thm:kernelBound}
\kernelBoundText
\end{theorem}

This theorem, together with the correctness of the reduction rules as proved in~\cite{kelk2016reduction}, immediately implies a linear kernel for \problemName.

To show how we prove the theorem, we will need to  introduce some terminology as we go.

A \emph{quartet} $Q$ is any set of $4$ elements in $\taxa$. 
If $T_1|_{Q} \neq T_2|_{Q}$, we say that $Q$ is a \emph{conflicting quartet} for $(T_1,T_2)$. 

As a crucial step we prove that for any $S$ large enough with respect to the degree of $S$ in both $T_1$ and $T_2$, either there exists a conflicting quartet or one of the reduction rules applies.

\newcommand{\smallDegreeSizeBoundText}{Let $S$ be a subset of $\taxa$ with $d_1$ the degree of $S$ in $T_1$, and $d_2$ the degree of $S$ in $T_2$.
If $|S|> 9(d_1+d_2)-12$, then either $T_1|_S \neq T_2|_S$ or one of Reduction Rules~\ref{rule:cherry} or~\ref{rule:chain} applies to $(T_1,T_2)$. In particular if $(T_1,T_2)$ is  irreducible under  Rules~\ref{rule:cherry} or~\ref{rule:chain} and $|S| \geq 9(d_1+d_2)-11$, then there exists a conflicting quartet $Q \subseteq S$, and such a quartet can be found in polynomial time.}
\begin{lemma}\label{lem:smallDegreeSizeBound}
\smallDegreeSizeBoundText
\end{lemma}


The next
result implies that if we have a large enough number of conflicting quartets that are also spanning-disjoint in both $T_1$ and $T_2$, then we are done. While it is intuitively clear that such quartets can be leveraged to create a high parsimony score in one tree, some care has to be taken to keep the parsimony score low in the other tree. 

\newcommand{\combiningQuartetsText}{
Let ${\cal Q} = \{Q_1, \dots, Q_k\}$ be a set of conflicting quartets for $T_1,T_2$, such that $Q_1, \dots Q_k$ are spanning-disjoint in $T_1$ and in $T_2$.

Then $\dmp(T_1,T_2) \geq k$, and we can find a witnessing character in polynomial time.}
\begin{lemma}\label{lem:combiningQuartets}
\combiningQuartetsText
\end{lemma}

In combination, Lemmas~\ref{lem:smallDegreeSizeBound} and~\ref{lem:combiningQuartets} allow us to show that $\dmp(T_1,T_2) \geq k$ providing that we can find at least $k$ sets $S_1, \dots S_k$ that are  spanning-disjoint in both trees and satisfy the conditions of Lemma~\ref{lem:smallDegreeSizeBound}.

We will find $k$ such sets as part of the construction of a character that witnesses $\dmp(T_1,T_2) \geq k$, for any reduced instance with $|\taxa| \geq \alpha k$.
In order to construct this character, we first create a partition of $\taxa$ into large subsets, as described by the following lemma.

\newcommand{\bigPartitionText}{
Suppose that $|\taxa| \geq 2ct$ for some integers $c$ and $t$, and let $T_1$ be a phylogenetic tree on $\taxa$.

Then in polynomial time we can construct a partition $S_1, \dots, S_t$ of $\taxa$ with $S_1, \dots, S_t$ spanning-disjoint in $T_1$,
such that $|S_i| \geq c$ for each $i$.}
\begin{lemma}\label{lem:bigPartition}
\bigPartitionText
\end{lemma}

We note that there is a one-to-one correspondence between partitions and characters on $\taxa$, in the following sense.
Given a partition $S_1,\dots S_t$ of $\taxa$, we may define a character $\chi:\taxa \rightarrow[t]$ such that $\chi(x) = i$ if $x \in S_i$, for each $i \in [t]$. Call such a character the character \emph{defined} by $S_1,\dots S_t$.

Thus let us consider the character $\chi$ on $\taxa$ defined by the partition described by Lemma~\ref{lem:bigPartition}. Since $S_1,\dots S_t$ are spanning-disjoint in $T_1$, Lemma~\ref{lem:PSlowerBound} tells that the parsimony score of $T_1$ with respect to $\cha$ is exactly $t-1$.

\newcommand{\wellBehavedSetsText}{Let $\cha$ be the character defined by the partition $S_1, \dots, S_t$ where $S_1, \dots, S_t$ are spanning-disjoint in $T_1$,
and assume 
$$t \geq 
\lceil\frac{(2d_1d_2+d_1)}{d_1d_2 -d_1-d_2}\rceil k$$ 
Then either 
${\dmp}_{\cha}(T_1,T_2) \geq k$,
or in polynomial time we can find a set of indices $i_1, \dots i_{k'}$ with $k' \geq k$ such that:
\begin{itemize}
    \item  $S_{i_1}, \dots S_{i_{k'}}$ are spanning-disjoint in $T_2$ (as well as $T_1$);
    \item each $S_{i_j}$ has degree at most $d_1$ in $T_1$; and
    \item each $S_{i_j}$ has degree at most $d_2$ in $T_2$.
\end{itemize}}
\begin{lemma}\label{lem:wellBehavedSets}
\wellBehavedSetsText
\end{lemma}

We will prove Theorem~\ref{thm:kernelBound} by combining these results in the following way. Fix integers $d_1,d_2$ to be determined later.
Assume $(T_1,T_2)$ is irreducible under Reduction Rules~\ref{rule:cherry} and~\ref{rule:chain},
and assume that $|\taxa| \geq 2ct$, where $c = 9(d_1+d_2)-11$ and $t \geq \lceil\frac{(2d_1d_2+d_1)}{d_1d_2 -d_1-d_2}\rceil k$ 
(this holds if $|\taxa| \geq \alpha k$).
By Lemma~\ref{lem:bigPartition}, there exists a partition $S_1, \dots S_t$ of $\taxa$ with $S_1, \dots S_t$ spanning-disjoint in $T_1$ and $|S_i| \geq c$ for each $i\in [t]$.
Let $\cha$ be the character defined by this partition. 
If ${\dmp}_{\cha}(T_1,T_2) \geq k$, we may return $\cha$.
Otherwise, we may apply Lemma~\ref{lem:wellBehavedSets} to get a set of indices  $i_1, \dots i_k$ such that  $S_{i_1}, \dots S_{i_k}$ are spanning-disjoint in $T_2$ (as well as in $T_1$),  each $S_{i_j}$ has degree at most $d_1$ in $T_1$, and  each $S_{i_j}$ has degree at most $d_2$ in $T_2$.
But then each $S_{i_j}$ satisfies the conditions of Lemma~\ref{lem:smallDegreeSizeBound}, and therefore for each $j \in [k]$ there exists  a conflicting quartet $Q_j \subseteq S_{i_j}$.
Moreover, as $S_{i_1}, \dots S_{i_k}$ are spanning-disjoint in $T_1$ and $T_2$, the quartets $Q_1, \dots Q_k$ are also spanning-disjoint in  $T_1$ and $T_2$.
Then Lemma~\ref{lem:combiningQuartets} implies that $\dmp(T_1,T_2) \geq k$.

By setting $d_1 = 4$ and $d_2 = 5$, we get that $\alpha = 560$, giving the desired bound.

In the next subsections we prove each of these lemmas, and then the main theorem, in turn.

\subsection{Reduction Rules}

We begin by stating the reduction rules for our kernelization result.

\begin{krule}\label{rule:cherry}[Cherry reduction rule]
If there exist $x,y \in \taxa$ such that in each of $T_1,T_2$ there exists an internal vertex $u$ adjacent to both $x$ and $y$, then replace $(T_1, T_2)$ with $(T_1|_{\taxa\setminus\{x\}}, T_2|_{\taxa\setminus\{x\}})$.
\end{krule}

\begin{krule}\label{rule:chain}[Chain reduction rule]
Suppose that there exists a sequence of leaves $x_1, \dots x_r \in \taxa$ with $r \geq 5$, such that in both $T_1$ and $T_2$, there exists a path of internal vertices $p_1,\dots, p_r$ (possibly with $p_1=p_2$ and possibly with $p_{r-1}=p_r$\steven{)}, such that for each $i \in [r]$ $p_i$ is the internal vertex adjacent to $x_i$.
Then replace $(T_1,T_2)$ with $(T_1|_{\taxa\setminus\{x_5, \dots, x_r\}}, T_2|_{\taxa\setminus\{x_5, \dots x_r\}})$
(thus, the common chain is reduced to length 4).
\end{krule}

The correctness of these rules was previously proved in~\cite{kelk2016reduction}. 
\begin{theorem}\label{thm:dmpReduction}
Let $(T_1',T_2')$ be an instance of \problemName derived from $(T_1, T_2)$ by an application of Reduction Rules~\ref{rule:cherry} or~\ref{rule:chain}.
Then $\dmp(T_1', T_2') = \dmp(T_1, T_2)$.
\end{theorem}
Correctness of the chain reduction rule follows from Theorem 3.1 in~\cite{kelk2016reduction}. Correctness of the cherry reduction rule follows as a subcase of Theorem 4.1 in~\cite{kelk2016reduction} (in particular, the cherry reduction is an instance of the ``traditional" case of the generalized subtree reduction from~\cite{kelk2016reduction}, where the subtree has $2$ leaves).

Our main contribution is to show that if an instance is reduced by these rules then its size is bounded by a linear function of $\dmp$.






\subsection{Small degree sets}



In this section we prove Lemma~\ref{lem:smallDegreeSizeBound}.

\newtheorem*{BRsmallDegreeSizeBound}{Lemma~\ref{lem:smallDegreeSizeBound}}
\begin{BRsmallDegreeSizeBound}
\smallDegreeSizeBoundText
\end{BRsmallDegreeSizeBound}

\begin{proof}
Since unrooted binary trees are characterized by their quartets \cite[Theorem 6.3.5(iii)]{SempleSteel2003} the last statement of the theorem follows directly. 

We will show that if $T_1|_S = T_2|S$ and neither of the reduction rules applies to $(T_1,T_2)$, then $|S| \leq 9(d_1+d_2)-12$. This implies the main claim of the lemma.
Let us denote $T|_S=T_1|_S=T_2|_S$.

Consider the \emph{backbone} graph of $T|_S$ obtained by deleting all leaves. Let $P_C$ be the set of nodes having degree 1 on the backbone, which we refer to as \emph{parents} of a cherry in $T|_S$. Let $P_L$ be the set of nodes having degree 2 on the backbone, which we refer to as \emph{parents} of a leaf of $T|_S$. All remaining vertices on the backbone have degree 3. 
Thus $|S|$, the total number of leaves of $T|_S$ is $2|P_C|+|P_L|$. We call the path between any two odd degree vertices on the backbone, having internal nodes only in $P_L$, a \emph{side} of the backbone. 

First notice that for each cherry in $T|_S$, there must exist in $T_1[S]$, the spanning tree on $S$ in $T_1$, or in $T_2[S]$ a node, incident to a pending edge, between at least one of its two leaves and its corresponding node in $P_C$. Otherwise Reduction Rule 1 can be applied. In particular this implies that $|P_C|\leq d_1+d_2$. 

Thus at least $P_C$ of the $d_1+d_2$ pending edges must be used for ``cutting'' the cherries, each of them cutting 1 leaf of a cherry. Let us choose one such leaf from each cherry, and call these the \emph{cut-leaves}.

After removing cut-leaves,
every node in $P_C$ and $P_L$ is now the parent of 1 leaf in $T|_S$. Every side of the backbone contains at most 4 vertices in $P_C$ and $P_L$, unless $T_1[S]$ or $T_2[S]$ has a node of a pending edge or a node adjacent to a node of a pending edge on that side. 
We show
that every such pending edge on a side may increase the number of $P_L$-nodes on that side by at most $5$ (see Figure~\ref{fig:backboneExample}). 
Indeed, suppose a side of the backbone has in total $d$ pending edges in both $T_1$ and $T_2$, but more than $4 + 5d$ nodes in $P_L$, i.e. at least $5(d+1)$.
Then $T|_{S}$ contains a chain of length $5(d+1)$, which we can split up into $d+1$ chains of length $5$. Clearly at least one of these chains has no pending edge in either $T_1$ or $T_2$, and so $T_1,T_2$ have a common chain of length $5$, a contradiction.

Thus the total number of nodes from $P_C$ and $P_L$ on a side is at most five times the number of pending edges (in $T_1[S]$ or $T_2[S]$) on that side, plus $4$.
Otherwise Reduction Rule 2 can be applied. Given that we already used $|P_C|$ pending edges for cutting the cherries, we have $d_1+d_2-|P_C|$ pending edges left to be distributed over the sides. 

The number of sides on the backbone is the number of edges in an unrooted binary tree with $|P_C|$ leaves, which is $2|P_C|-3$. Therefore the total number of leaves of $T|_S$ is \begin{eqnarray*}|S|=2|P_C|+|P_L| & \leq & |P_C| + 4(2|P_C|-3)+5(d_1+d_2-|P_C|) \\
& \leq & 4|P_C|+5(d_1+d_2)-12.
\end{eqnarray*}
Clearly, this attains its largest value if $|P_C|=d_1+d_2$, in which case $|S|\leq 9(d_1+d_2)-12$, as was to be proven. 
\end{proof}

\begin{figure}
    \begin{subfigure}[b]{0.9\textwidth}
    \includegraphics[width=\textwidth]{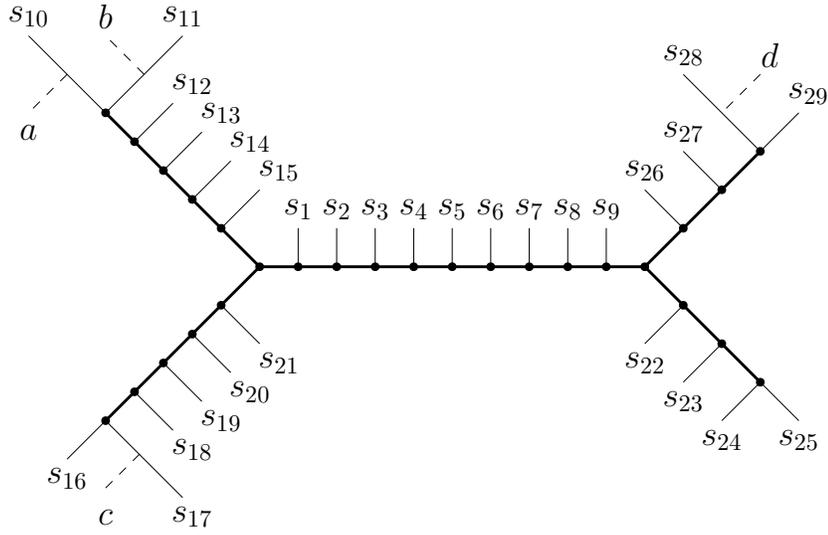}
    \caption{Backbone of $T_1|_S$ within $T_1$.}
    \end{subfigure}
    \begin{subfigure}[b]{0.9\textwidth}
    \includegraphics[width=\textwidth]{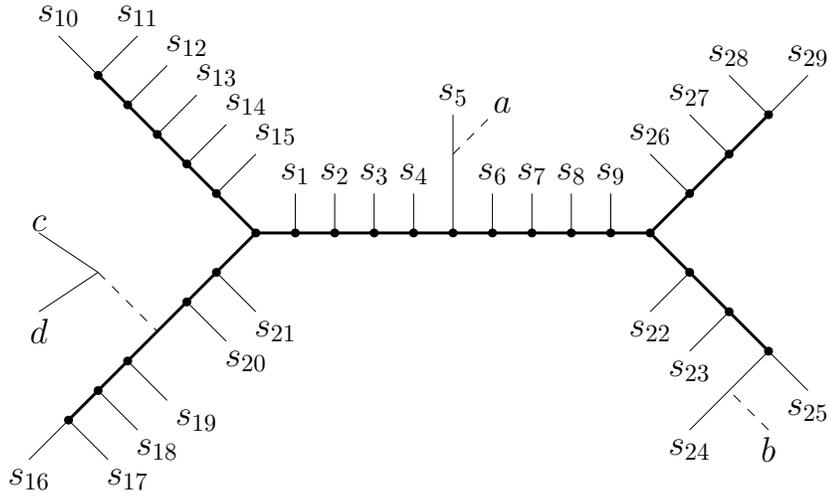}
    \caption{Backbone of $T_2|_S$ within $T_2$}
    \end{subfigure}
    \caption{Example illustration of the backbone of $T_1|_S = T_2|_S$ within $T_1$ and $T_2$, where $S = \{s_1, \dots, s_{29}\}$. Edges and vertices of the backbone are in bold. 
    Observe that $T[S]$ has the chain $s_1, \dots, s_9$, but $(T_1,T_2)$ do not have a common chain of length greater than $4$, as the leaf $s_5$ has a sibling $a$ in $T_2$.
    } 
    \label{fig:backboneExample}
\end{figure}

\subsection{Combining conflicting quartets}

In this section we prove Lemma~\ref{lem:combiningQuartets}.

\newtheorem*{BRcombiningQuartets}{Lemma~\ref{lem:combiningQuartets}}

\begin{BRcombiningQuartets}
\combiningQuartetsText
\end{BRcombiningQuartets}
\begin{proof}
For a quartet $Q$ and tree $T$, we say that $T|_Q = ab|cd$ if $Q = \{a,b,c,d\}$ and in $T$ the path between $a$ and $b$ is edge-disjoint from the path between $c$ and $d$.
Without loss of generality, we may assume $Q_i = \{a_i,b_i,c_i,d_i\}$,  $T_1|_{Q_i} = a_ib_i|c_id_i$ and $T_2|_{Q_i} = a_ic_i|b_id_i$ for each $i \in [k]$.

We will show how to build a character $\cha$ with two states, such that $\PS{\cha}{T_1} \leq k$, and $\PS{\cha}{T_2} \geq 2k$. This shows that 
${\dmp}_{\cha}(T_1,T_2) \geq k$,
as required.

The idea is to construct $\cha$ in such a way that, for each quartet $Q_i$, $\cha(a_i)=\cha(b_i)\neq \cha(c_i)=\cha(d_i)$. This will ensure that $\PS{\cha}{T_2}$ is at least $2k$, as $T_2$ will have at least $2k$ edge-disjoint paths (from $a_i$ to $c_i$ and from $b_i$ $d_i$, for each $i \in [k]$) that each require at least one change in state along some edge.

For each $Q_i$, let $e_{Q_i}$ denote an edge in $T_1$ such that in $T_1[Q_i]$, $e_i$ is on the path that separates $\{a_i,b_i\}$ from $\{c_i,d_i\}$.

Now we construct a function $\ext:V(T_1) \rightarrow \{\textsc{red}, \textsc{blue}\}$ as follows.
Start by choosing an arbitrary leaf in $T_1$, say without loss of generality $a_1$, and set $\ext(a_1) = \textsc{red}$.
Now proceed as follows. 
For any edge $uv$ in $T_1$ such that $\ext(u)$ is defined but $\ext(v)$ is not, we set $\ext(v) = \ext(u)$, unless $uv = e_{Q_i}$ for some $i$. In that case, we set $\ext(v) = \textsc{blue}$ if $\ext(u) = \textsc{red}$, and set  $\ext(v) = \textsc{red}$ otherwise. 
Now we can 
let
$\chi$ be the restriction of $\ext$ to $\taxa$.

By construction, $\ext$ is an extension of $\cha$ to $T_1$ and $\extScore(\ext) = |e_{Q_i}: i \in [k]| =k$. This is enough to show that $\PS{\cha}{T_1} \leq k$.
We now show that $\cha(a_i) = \cha(b_i) \neq \cha(c_i) = \cha(d_i)$, for each $i \in [k]$.
To see this, consider the spanning tree $T_1[Q_i]$. By construction, $T_1[Q_i]$ contains the edge $e_{Q_i}$ and $e_{Q_i}$ separates $\{a_i,b_i\}$ from $\{c_i,d_i\}$. Let $u_i,v_i$ be the vertices of $e_{Q_i}$, with $u_i$ the vertex closer to $a_i$ and $b_i$. Note that $T_1[Q_i]$ cannot contain $e_{Q_j}$ for any $j\neq i$, as $T_1[Q_i]$ and $T_1[Q_j]$ are edge-disjoint. It follows that $u_i,a_ib_i$ are all assigned the same value by $\ext$ and $v_i,c_i,d_i$ are assigned the opposite value.
Thus by definition of $\cha$, we have $\cha(a_i) = \cha(b_i) = \ext(u_i) \neq \ext(v_i) = \cha(c_i) = \cha(d_i)$.

It remains to observe that as $Q_1, \dots Q_k$ are spanning-disjoint in $T_2$, the $a_i-c_i$ and $b_i-d_i$ paths in $T_2$ are pairwise edge-disjoint for all $i \in [k]$. Then as $\cha(a_i)\neq \cha(c_i)$ and $\cha(b_i) \neq \cha(d_i)$, there exist at least $2k$ edges $uv$ in $T_2$ with $\ext_2(u) \neq \ext_2(v)$, for any extension $\ext_2$ of $\cha$ to $T_2$. It follows that $\PS{\cha}{T_2} \geq 2k$, and so $\dmp(T_1,T_2) \geq {\dmp}_{\cha}(T_1,T_2) =  |\PS{\cha}{T_1} - \PS{\cha}{T_2}| \geq 2k-k = k$.

Since each edge is processed at most once in the construction of $\cha$, it is clear that this construction takes polynomial time.
\end{proof}

\subsection{Constructing an initial partition}

In this section we prove Lemma~\ref{lem:bigPartition}. 

\newtheorem*{BRbigPartition}{Lemma~\ref{lem:bigPartition}}
\begin{BRbigPartition}
\bigPartitionText
\end{BRbigPartition}

\begin{proof}
We prove the claim by induction on $t$. For the base case, if $t = 1$ then we may let $S_1 = \taxa$, and we have the desired partition.

For the inductive step, assume $|\taxa| \geq 2ct$ and that the claim is true for smaller values of $t$.
We first fix an arbitrary rooting on $T_1$. That is, choose an arbitrary edge $e$ in $T_1$ and subdivide it with a new (temporary) vertex $r$, then orient all edges in $T_1$ away from $r$.
Under this rooting, let $u$ be a lowest vertex in $T_1$ for which $u$ has at least $c$ descendants in $\taxa$. Let $S_t \subseteq \taxa$ be the set of these descendants,
Note that since $T_1$ is binary, $|S_t| < 2c$, as otherwise one of the two children of $u$ would be a lower vertex with at least $c$ descendants. 

Now consider the induced subtree $T_1|_{\taxa'}$, where $\taxa' = \taxa \setminus S_t$. As $|S_t| < 2c$, we have $\taxa' \geq 2c(t-1)$. Then by the inductive hypothesis, we can construct a partition $S_1, \dots, S_{t-1}$ of $\taxa'$ with $S_1, \dots, S_{t-1}$ spanning-disjoint in $T_1|_{\taxa'}$,
such that $|S_i| \geq c$ for each $i$. By construction it is clear that $S_t$ is spanning-disjoint in $T_1$ from  $S_1, \dots, S_{t-1}$. Thus $S_1, \dots, S_t$ is the desired partition.

As the construction of $S_t$ can be done in polynomial time and this process is repeated $t \leq |\taxa|$ times, the entire process takes polynomial time.
\end{proof}

\subsection{Well-behaved sets}

In this section we prove Lemma~\ref{lem:wellBehavedSets}. 
We start with an observation:

 
 \begin{observation}\label{obs:treeDegrees}
  For any (not necessarily binary) unrooted tree $T$ with $n$ vertices, and any integer $d \geq 1$, the number of vertices in $T$ with degree strictly greater than $d$ is at most $n/d$.\footnote{The proof of this observation is based on an argument in \cite{SE2013}.} 
 \end{observation}
 \begin{proof}
 For each vertex $v$ in $T$ let $d(v)$ denote the degree of $v$. 
 Recall that an unrooted tree with $n$ vertices has exactly $n-1$ edges. It follows that $\sum_{v \in V(T)}d(v) = 2|E(T)| = 2n-2$.
 Now suppose that $T$ has $m > n/d$ vertices with degree strictly greater than $d$, i.e. at least $d+1$. The remaining $n-m$ vertices all have degree at least $1$, from which it follows that $\sum_{v \in V(T)}d(v) \geq m(d+1) + n-m = 
md + n \geq (n/d)d + n = 2n$, a contradiction.
 \end{proof}

\newtheorem*{BRwellBehavedSets}{Lemma~\ref{lem:wellBehavedSets}}
\begin{BRwellBehavedSets}
\wellBehavedSetsText
\end{BRwellBehavedSets}

\begin{proof}
 By Lemma~\ref{lem:PSlowerBound},
 $\PS{\cha}{T_1} = t-1$.
 If $\PS{\cha}{T_2} \geq t+k-1$, then 
${\dmp}_{\cha}(T_1,T_2)\geq k$
 as required. 
 So we may assume that $\PS{\cha}{T_2} \leq t+k-2$. 
 Let $\delta = \PS{\cha}{T_2} - \PS{\cha}{T_1}$, and observe that $0 \leq \delta \leq k-1$. 
 

 
 We now construct a partition  $P_1,\dots P_s$ of $\taxa$ which is spanning-disjoint in $T_2$ 
 (see Figure~\ref{fig:T2partitionFromT1partition} for an illustration).
 Let $\ext_2$ be an optimal extension of $\cha$ to $T_2$.
 As $\PS{\cha}{T_2} = \PS{\cha}{T_1} + \delta = t + \delta - 1$, the forest induced by $\ext_2$ has exactly $s$ monochromatic connected components, where $s = t + \delta$.
 Let $P_1, \dots, P_s$ be the partition of $\taxa$ formed by taking the intersection of $\taxa$ with the vertex set of each tree in this forest.
 Observe that by construction $P_1,\dots P_s$ are spanning-disjoint in $T_2$, and that furthermore each $P_j$ is a subset of $S_i$ for some $i \in [t]$ (as each element of $P_j$ is assigned the same value by $\ext_2$, and thus by $\cha$).

\begin{figure}
    \begin{subfigure}[b]{0.5\textwidth}
    \includegraphics[width=0.9\textwidth]{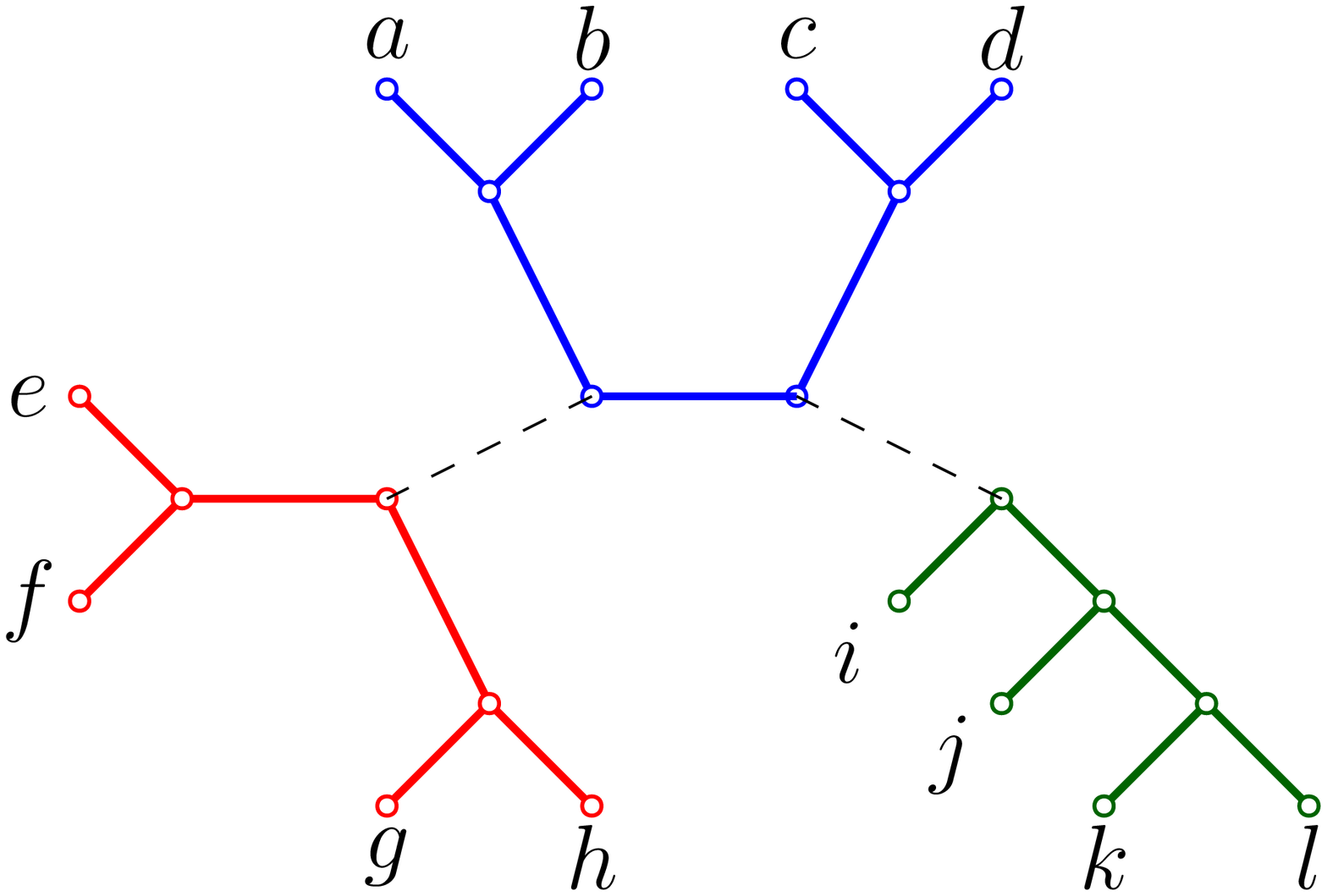}
    \caption{Partition in $T_1$\\
    $S_1=\{a,b,c,d\}$,\\ $S_2=\{e,f,g,h\}$,\\ $S_3=\{i,j,k,l\}$.}
    \end{subfigure}
    \begin{subfigure}[b]{0.5\textwidth}
    \includegraphics[width=0.9\textwidth]{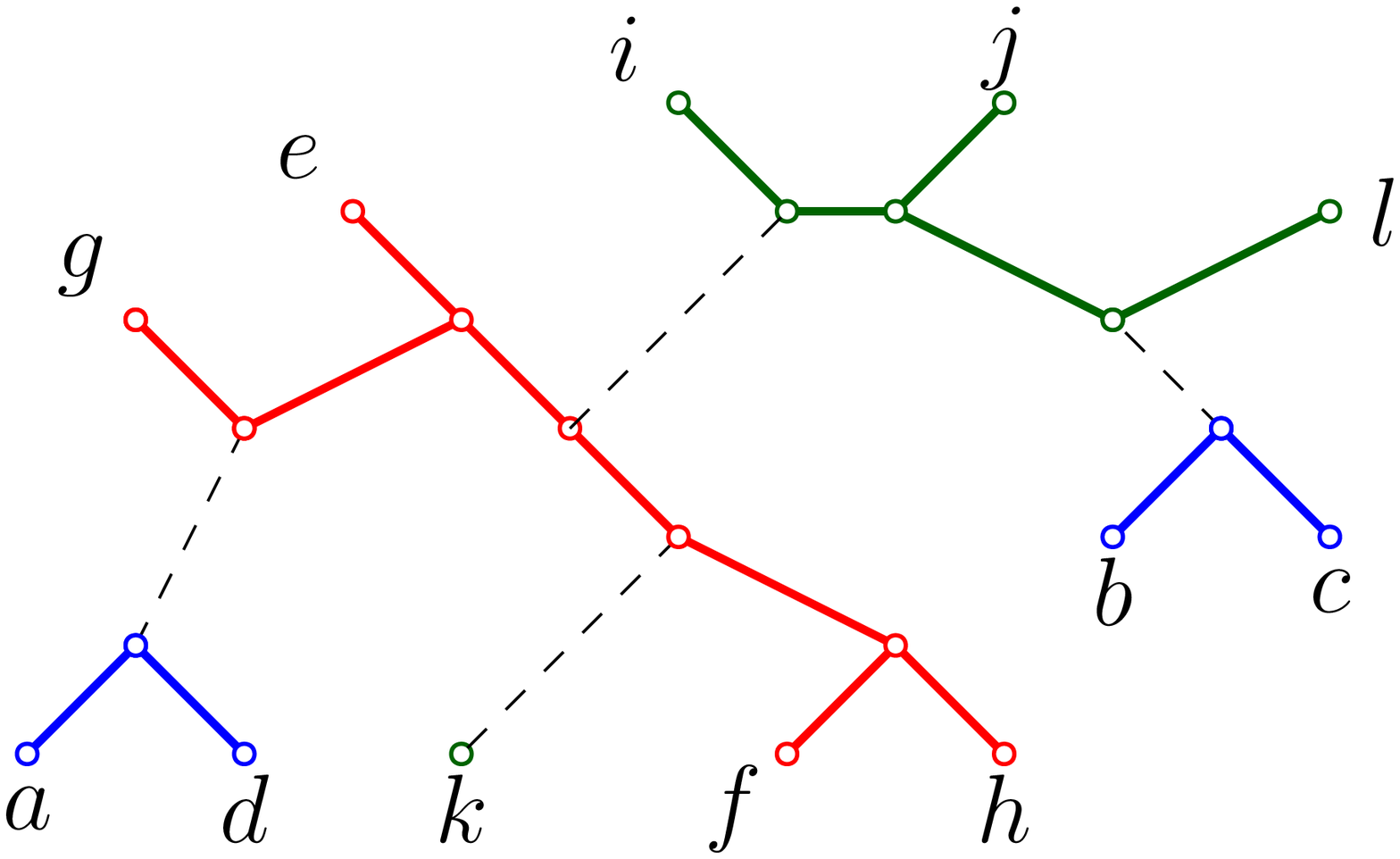}
    \caption{Partition in $T_2$ \\
    $P_1 = \{a,d\}$, $P_2 = \{b,c\}$,\\ $P_3 = \{e,f,g,h\}$,\\ $P_4 = \{i,j,l\}$, $P_5 = \{k\}$.}
    \end{subfigure}
    \caption{Illustration of the construction of partition $P_1,P_2,P_3,P_4,P_5$ from $S_1,S_2,S_3$.
    } 
    \label{fig:T2partitionFromT1partition}
\end{figure}

Now let ${\cal I} \subseteq [t]$ denote the set of indices $i$ in $[t]$ such that 
\begin{itemize}
    \item $S_i = P_j$ for some $j \in [s]; $;
    \item $S_i$ has degree at most $d_1$ in $T_1$; and
    \item $S_i$ has degree at most $d_2$ in T$_2$.
\end{itemize}

Note that since $P_1, \dots P_j$ are spanning-disjoint in $T_2$, the sets $\{S_i: i \in {\cal I}\}$ are also spanning-disjoint in $T_2$.
Notice that it is sufficient to prove that $|{\cal I}| \geq k$, whence any subset of $k$ indices from ${\cal I}$ satisfies the lemma.
We will prove this by providing upper bounds on the number of indices in $[t]$ that do not satisfy the conditions of ${\cal I}$.

Let ${\cal I}_0$ denote the set of indices $i \in [t]$ such that $P_j \neq S_i$ for any $j \in [s]$. 
We first claim that  $|{\cal I}_0| \leq \delta$.
Indeed, since every $P_j$ is a subset of some $S_i$ and $S_1, \dots S_t$ and $P_1,\dots, P_s$ are both partitions of $\taxa$, we have that for every $i \in {\cal I}_0$, there exist at least two distinct indices $j,j' \in [s]$ for which $P_j, P_{j'}\subset S_i$. Hence, $s \geq 2|{\cal I}_0| + |[t]\setminus {\cal I}_0| = t + |{\cal I}_0|$. Therefore if $|{\cal I}_0| > \delta$ then $s > t+\delta$, contradicting the definition of $s$.
Thus, we have  $|{\cal I}_0| \leq \delta$.

Next, let ${\cal I}_{>d_1}$ denote the set of indices $i \in [t]$ for which $S_i$ has degree greater than $d_1$ in $T_1$.
We will show that
$|{\cal I}_{>d_1}| \leq t/d_1$.
For each $i \in [t]$, compress the spanning subtree $T_1[S_i]$ to a single vertex, and observe that the degree of this vertex is equal to the degree of $S_i$ in $T_1$. 
Any vertex $u$ which is not part of any 
$T_1[S_i]$ is merged with
one of its neighbours.
Note that this merging process can only increase the degrees of the remaining vertices. Call the resulting tree $T_1'$. See Figure~\ref{fig:T1primeConstruction}.
$T'_1$ has $t$ vertices, each of them corresponding to a subset $S_i$, and having degree at least the degree of the corresponding $S_i$ in $T_1$.
Now by Observation~\ref{obs:treeDegrees}, there are at most $t/d_1$ vertices in $T'_1$ with degree greater than $d_1$. It follows that there are at most $t/d_1$ values of $i \in [t]$ for which $S_i$ has degree greater than $d_1$ in $T_1$, and thus $|{\cal I}_{>d_1}| \leq t/d_1$ as we wanted to show.

\begin{figure}
    \begin{subfigure}[b]{0.5\textwidth}
    \includegraphics[width=0.9\textwidth]{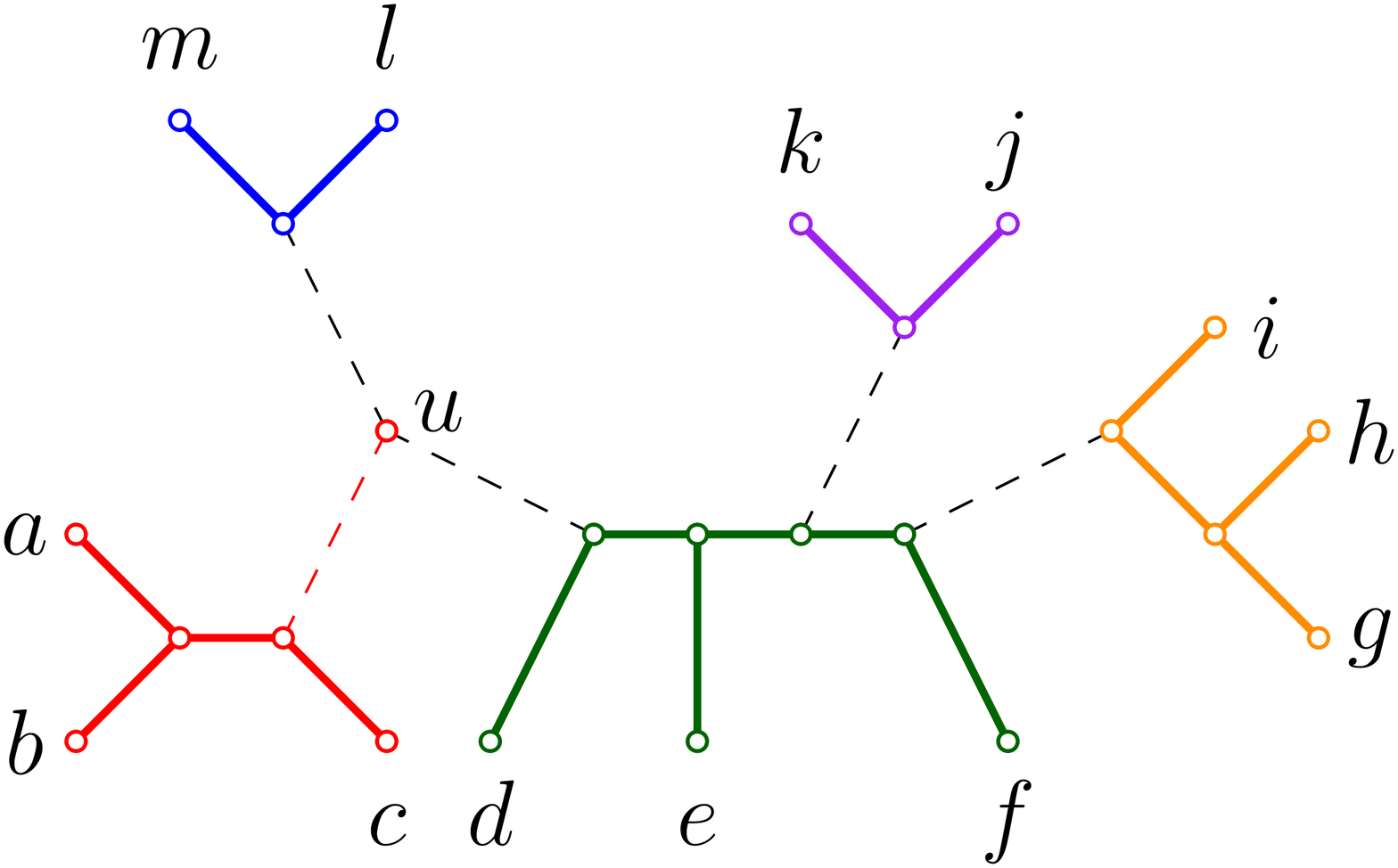}
    \caption{$T_1$}
    \end{subfigure}
    \begin{subfigure}[b]{0.5\textwidth}
    \includegraphics[width=0.9\textwidth]{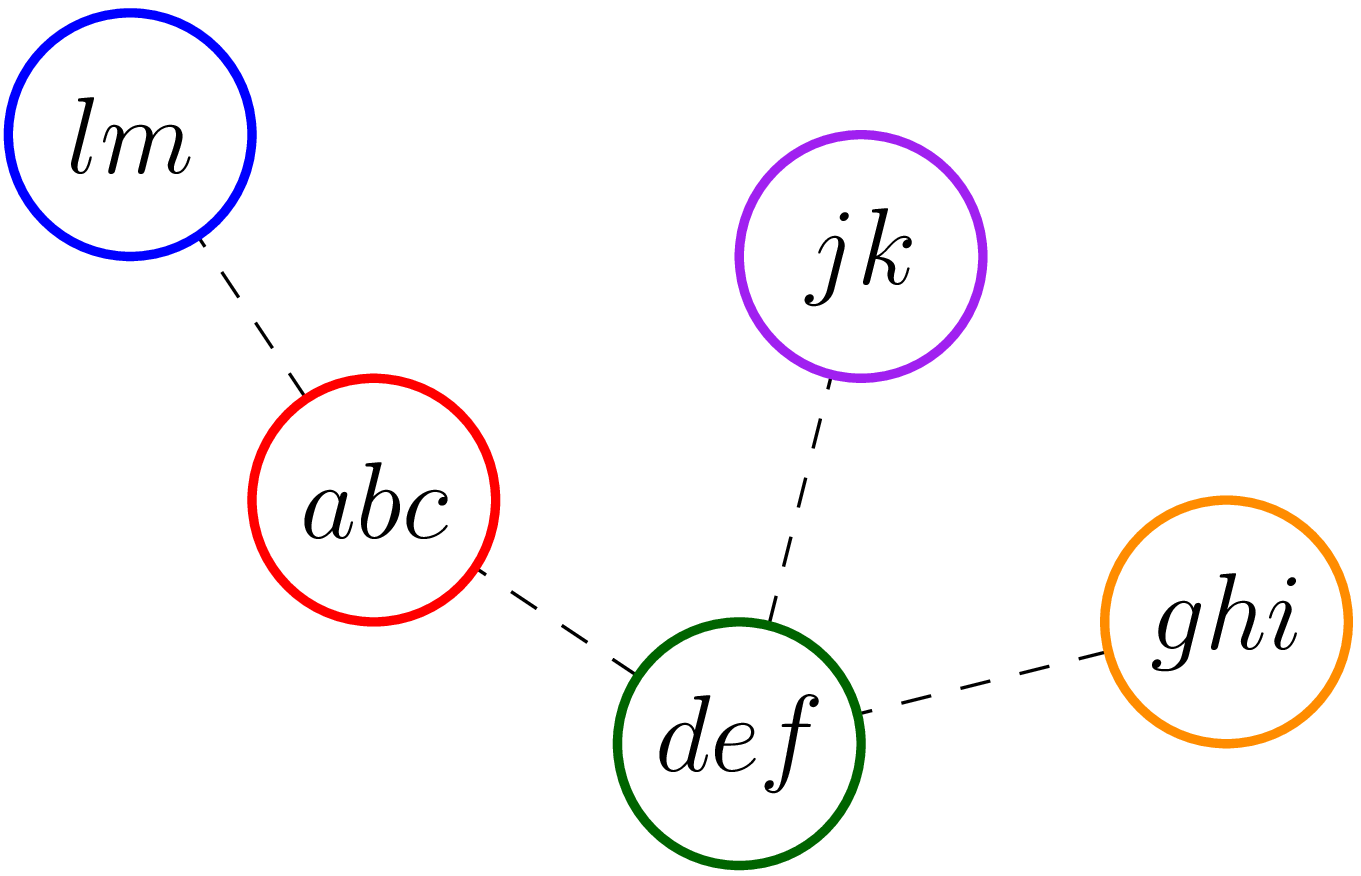}
    \caption{$T_1'$}
    \end{subfigure}
    \caption{Illustration of the construction of auxiliary tree $T_1'$, given a partition of $\taxa$ with $S_1 = \{a,b,c\}$, $S_2 = \{d,e,f\}$, $S_3 = \{g,h,i\}$, $S_4 = \{j,k\}$, $S_5 = \{l,m\}$.
Note that the internal vertex labelled $u$ is not part of $T_1[S_i]$ for any $i$, so we merge it with an arbitrary adjacent vertex. In this case we merge $u$ into $S_1 = \{a,b,c\}$, which is why $S_1$ has degree $1$ in $T_1$ but degree $2$ in $T_1'$.
    } 
    \label{fig:T1primeConstruction}
\end{figure}

Similarly 
let  ${\cal J}_{>d_2}$ denote the set of indices $j \in [s]$ for which $P_j$ has degree greater than $d_2$ in $T_2$.
By similar arguments as used for ${\cal I}_{>d_1}$ above, we can show that $|{\cal J}_{>d_2}| \leq s/d_2$.
 
Notice that for any $i\in [t]$, if $i$ is not in ${\cal I}$, then either $i \in {\cal I}_0$, or $i \in {\cal I}_{>d_1}$, or there exists $j \in {\cal J}_{>d_2}$ such that $S_i = P_j$.
 We therefore have that $|{\cal I}| \geq t - |{\cal I}_0| - |{\cal I}_{>d_1}| - |{\cal J}_{>d_2}|
 \geq t - \delta - t/d_1 - s/d_2$.
 
Now, using that $t \geq \frac{  (2d_1d_2+d_1)}{d_1d_2 -d_1-d_2} k$, $s = t+\delta$ and $\delta \leq k-1$, we have:

\begin{align*}
|{\cal I}| & \geq t - |{\cal I}_0| - |{\cal I}_{>d_1})| -  |{\cal J}_{>d_2}| \\
 & \geq t - \delta - t/d_1 - s/d_2  \\
 & = t - \delta - t/d_1 - (t+\delta)/d_2  \\
  & = \frac{d_1d_2t - d_1d_2\delta - d_2t - d_1t -d_1\delta }{d_1d_2} \\
  & = \frac{(d_1d_2 -d_1-d_2)t - (d_1d_2+d_1)\delta}{d_1d_2} \\
  & \geq \frac{(d_1d_2 -d_1-d_2)t - (d_1d_2+d_1)(k-1)}{d_1d_2} \\
  & \geq \frac{ (2d_1d_2+d_1)k - (d_1d_2+d_1)(k-1)}{d_1d_2} \\
  & =\frac{ d_1d_2k +d_1d_2 + d_1}{d_1d_2} \\
  & > \frac{d_1d_2k}{d_1d_2} \\
   & = k,
\end{align*}

\noindent as we needed to prove. 
To see that ${\cal I}$ can be constructed in polynomial time, it suffices to observe that the partition $P_1, \dots, P_s$ can be constructed in polynomial time (as the $\ext_2$ can be found in polynomial time), and after this each $S_i$ can be checked for membership in ${\cal I}$ in polynomial time.
 %
\end{proof}

\subsection{Proof of Theorem~\ref{thm:kernelBound}}

\begin{lemma}\label{lem:kernelBound}
Let $d_1,d_2$ be positive integers such that $d_1d_2 - d_1 - d_2 > 0$.
Let $(T_1,T_2)$ be a pair of binary unrooted phylogenetic trees on $\taxa$ that are irreducible under Reduction Rules~\ref{rule:cherry} and~\ref{rule:chain}.

Then if $|\taxa| \geq 2ct$, 
where $c = 9(d_1+d_2) - 11$
and $t = \lceil\frac{(2d_1d_2 + d_1)}{d_1d_2 -d_1-d_2}\rceil k$,
it holds that $\dmp(T_1,T_2) \geq k$, and we can find a witnessing character in polynomial time.
\end{lemma}
\begin{proof}
By Lemma~\ref{lem:bigPartition}, there exists a partition $S_1, \dots S_t$ of $\taxa$, all spanning-disjoint in $T_1$, and with $|S_i| \geq c$ for all $i\in [t]$.
Let $\cha$ be the character defined by $S_1, \dots, S_t$. If $\cha$ is a witness to $\dmp(T_1,T_2) \geq k$, then we may return $\cha$ and we are done.
Otherwise, 
we may apply Lemma~\ref{lem:wellBehavedSets} to find indices $i_1, \dots i_k$ such that:
\begin{itemize}
    \item $S_{i_1},  \dots S_{i_k}$ are all spanning-disjoint in $T_2$ (as well as in $T_1$);
    \item each $S_{i_j}$ has degree at most $d_1$ in $T_1$; and
    \item each $S_{i_j}$ has degree at most $d_2$ in $T_2$.
\end{itemize}

Now for each $S_{i_j}$, we have that $S_{i_j}$ has degree  $d^j_1 \leq d_1$ in $T_1$ and $d^j_2 \leq d_2$ in $T_2$, that $|S_{i_j}| \geq c > 9(d_1+d_2) - 11 \geq 9(d^j_1+d^j_2) - 11$ , and that $(T_1,T_2)$ is irreducible under Rules~\ref{rule:cherry} and~\ref{rule:chain}.
Thus we may apply Lemma~\ref{lem:smallDegreeSizeBound}, to find a conflicting quartet $Q_j \subseteq S_{i_j}$ for each $i_j$.

Finally, as $S_{i_1},  \dots S_{i_k}$ are spanning-disjoint in both $T_1$ and $T_2$, and as each $Q_j$ is a subset of $S_{i_j}$, we have that $Q_1, \dots, Q_k$ are also spanning-disjoint in both $T_1$ and $T_2$.
Therefore we may apply Lemma~\ref{lem:combiningQuartets} to find a witnessing character for $\dmp(T_1,T_2) \geq k$.
As each step of this process takes polynomial time, the construction of a witnessing character takes polynomial time.
\end{proof}

It remains to complete the proof of Theorem~\ref{thm:kernelBound}. 

\newtheorem*{BRkernelBound}{Theorem~\ref{thm:kernelBound}}
\begin{BRkernelBound}
\kernelBoundText
\end{BRkernelBound}
\begin{proof} 
The proof boils down to choosing the appropriate values of $d_1$ and $d_2$ such that $2ct = (9(d_1+d_2) - 11) \cdot  \lceil\frac{(2d_1d_2 + d_1)}{d_1d_2 -d_1-d_2}\rceil k = \alpha k$. 
For $d_1 = 4, d_2 = 5$ we get $c = 70$ and $t = 4k$, yielding the value of $\alpha = 560$ for $\alpha k = 2ct$.
\end{proof}
\noindent In the appendix, we show that $d_1 = 4, d_2 = 5$ is in fact the optimal choice of values for $d_1$ and $d_2$.



As a corollary to Theorem~\ref{thm:kernelBound} and 
Theorem~\ref{thm:dmpReduction},
we have that \problemName{} is fixed-parameter tractable with respect to $\dmp$. Specifically, the kernel can be solved using the exponential-time algorithm described in \cite{kelk2017note}, which computes the maximum parsimony distance of two trees on $n$ leaves in time  $O( 1.619^n \cdot \text{poly}(n))$.

\begin{corollary}\label{cor:dmpFPT}
\problemName{} has a kernel of size $\alpha k$, and can be solved in time 
$O(1.619^{\alpha k} \cdot poly(\alpha k) + poly(n))$,
with 
$k = \dmp(T_1,T_2)$.
\end{corollary}

For completeness,
we clarify that these results also prove that the decision problems
``$d_{MP} \leq k$?'', ``$d_{MP} \geq k$?'' and ``$d_{MP} = k$?'' can all be answered in time $f(k) \cdot poly(n)$. To answer ``$d_{MP} \leq k$?'', note that if the kernel has size at least $\alpha (k+1)$ the answer is definitely NO, and otherwise the algorithm from \cite{kelk2017note} can be applied to compute $d_{MP}$ directly; this can then be compared to $k$ to resolve the question. The ``$ d_{MP}\geq k$?'' question can be answered by asking ``$d_{MP} \leq k-1$?'' and negating the answer; and ``$d_{MP} = k$?'' can be answered by combining the answers to the $\leq$ and $\geq$ questions. 

\section{Corollaries: leveraging the kernel}
\label{sec:cor}

\subsection{A polynomial-time constant-factor approximation algorithm for \problemName{}}
\label{sec:approx}

We present how a constant factor approximation algorithm for \problemName{} can be designed using Theorem~\ref{thm:kernelBound} together with Reduction Rules~\ref{rule:cherry} and~\ref{rule:chain}.

In order to incorporate Reduction Rules~\ref{rule:cherry} and~\ref{rule:chain} into our approximation algorithm, we  require a way to construct a witnessing character for the original instance from a witnessing character for the reduced instance.

\begin{lemma}\label{lem:dmpReconstruction}
Let $(T_1',T_2')$ be an instance of \problemName{} derived from $(T_1, T_2)$ by an application of Reduction Rule~\ref{rule:cherry} or~\ref{rule:chain}, with $T_1',T_2'$ trees on $\taxa' \subset \taxa$.
Then given a character $\cha'$ on $X'$, we can derive a character $\cha$ on $X$ in polynomial time such that 
${\dmp}_{\cha}(T_1,T_2) \geq {\dmp}_{\cha'}(T_1',T_2')$.
\end{lemma}
\begin{proof}
First observe that by definition of the reduction rules, we may assume that $T_1' = T_1|_{\taxa'}$ and $T_2' = T_2|_{\taxa'}$ for some $\taxa' \subseteq \taxa$.
Assume without loss of generality that $\PS{\cha'}{T_2'} \geq \PS{\cha'}{T_1'}$, and let $\ext_1'$ be an optimal extension of $\cha'$ to $T_1'$. 
We will now define a function $\ext:V(T_1) \rightarrow \states$ such that $\ext(u) = \ext'(u)$ for all $u \in V(T_1')$, and
such that $\extScore_{T_1}(\ext_1) = \extScore_{T_1'}(\ext'_1) = \PS{\cha'}{T_1'}$.
Recall that $T_1|_{\taxa'}$ is derived from the spanning tree $T_1[\taxa']$ by suppressing vertices of degree $2$, and therefore $T_1[\taxa']$ can be derived from $T_1' = T_1|_{\taxa'}$ by repeatedly subdividing edges with degree-$2$ vertices.
Now construct $\ext_1$ as follows. For each vertex $v$ in $T_1'$, set $\ext_1(v) = \ext'_1(v)$. For every edge $e = uv$ that gets subdivided with one or more degree-$2$ vertices, set $\ext_1(u') = \ext_1'(u)$ for each such degree-$2$ vertex $u'$.
Thus, $\ext_1$ assigns a colour to every vertex in $T_1[\taxa']$, and by construction $\extScore_{T_1[\taxa']}(\ext_1) = \extScore_{T_1'}(\ext'_1)$.

In order to assign $\ext(v)$ to vertices $v$ of $T_1$ not in $T_1[\taxa']$, take any edge $e = uv$ in $T_1$ such that $\ext_1(u)$ has been assigned but $\ext_1(v)$ has not, and set $\ext_1(v) = \ext_1(u)$.
After completing this process, we have that $\ext_1$ assigns a colour to every vertex in $T_1$ (including its leaves) and $\extScore_{T_1}(\ext_1) = \extScore_{T_1'}(\ext'_1)$, as required. 

Now let the character $\cha$ be the restriction of $\ext_1$ to $\taxa$. Then by construction $\ext_1$ is an extension of $\cha$ on $\taxa$, whence $\PS{\cha}{T_1} \leq \extScore_{T_1}(\ext_1) = \PS{\cha'}{T_1'}$.
Moreover, we must have that $\PS{\cha}{T_1} \geq \PS{\cha'}{T_1'}$ and thus $\PS{\cha}{T_1} = \PS{\cha'}{T_1'}$. Indeed, if $\extScore_{T_1}(\ext) < \extScore_{T_1}(\ext_1)$ for some extension  $\phi$ of $\cha$ to $T_1$, then by considering the restriction of $\ext$ to $T_1[S]$, we can see that $\PS{\chi'}{T_1'} \leq \extScore_{T_1}(\ext) < \extScore_{T_1}(\ext_1)$, a contradiction as $\extScore_{T_1}(\ext_1) =\extScore_{T_1'}(\ext'_1) = \PS{\chi'}{T_1'}$.

Next we show that $\PS{\cha}{T_2} \geq \PS{\cha'}{T_2'}$.
Consider any optimal extension $\ext_2$ of $\cha$ to $T_2$, and take the restriction $\ext_2'$ of this function to $T_2'=T_2|_{\taxa'}$. Then clearly $\extScore_{T_2'}(\ext_2') \leq \extScore_{T_2}(\ext_2)$ and therefore
$\PS{\cha'}{T_2'} \leq \extScore_{T_2'}(\ext_2') \leq \extScore_{T_2}(\ext_2) = \PS{\cha}{T_2}$.


Thus we have  
${\dmp}_{\cha}(T_1,T_2)
\geq \PS{\cha}{T_2} - \PS{\cha}{T_1} \geq  \PS{\cha'}{T_2'} - \PS{\cha'}{T_1'} = 
{\dmp}_{\cha'}(T_1',T_2')$
\end{proof}

\begin{theorem}
For any positive integer $r$,
given an instance $(T_1,T_2)$ of \problemName, we can find in polynomial time a character $\cha$ such that 
$$1 \leq \frac{\dmp(T_1,T_2)}{{\dmp}_{\cha}(T_1,T_2)} \leq (1+1/r)\alpha$$
where $\alpha = 560$.
That is, \problemName\ has a constant factor approximation with approximation ratio $(1+1/r)\alpha$.
\end{theorem}
\begin{proof}
First apply Reduction Rules~\ref{rule:cherry} and~\ref{rule:chain} exhaustively, to derive an irreducible instance $(T_1',T_2')$.
By Theorem~\ref{thm:dmpReduction}, $\dmp(T_1',T_2') = \dmp(T_1,T_2)$.
Let $\taxa'$ be the leaf set of this reduced instance.
Now let $k$ be
the maximum integer such that $|\taxa'| \geq \alpha k$, where $\alpha = 560$.
If $k < r$, then  we can determine a character $\cha'$ for which ${\dmp}_{\cha'}(T_1',T_2') = \dmp(T_1',T_2')$ exactly in time  $O(1.619^{|X'|}\cdot poly(n)) = O(1.619^{\alpha r}\cdot poly(n))$ using the algorithm of~\cite{kelk2017note}.
Otherwise by Theorem~\ref{thm:kernelBound}, we can in polynomial time construct a character $\cha'$ on $\taxa'$ such that ${\dmp}_{\cha'}(T_1',T_2') \geq k$. 
In either case, by Lemma \ref{lem:dmpReconstruction}
we can extend $\cha'$ to a character $\cha$ on $\taxa$ such that ${\dmp}_{\cha}(T_1,T_2) \geq  {\dmp}_{\cha'}(T_1',T_2')\geq k$. We return $\cha$.


It remains to show that 
$\dmp(T_1,T_2)/(1+1/r)\alpha \leq 
{\dmp}_{\cha}(T_1,T_2)
\leq \dmp(T_1,T_2)$, from which the theorem follows.
The second inequality is by definition of $\dmp(T_1,T_2)$.
To see the first inequality:
if $k <r$ then by construction $\dmp(T_1,T_2) = \dmp(T_1',T_2') = {\dmp}_{\cha'}(T_1',T_2') \leq {\dmp}_{\cha}(T_1, T_2)  $. So now assume that $k \geq r$, and so by construction ${\dmp}_{\cha'}(T_1',T_2') \geq k \geq r$.
As stated in the preliminaries, the number of taxa provides an upper bound on the $\dmp$ of any instance. 
Thus, $\dmp(T_1',T_2') \leq |\taxa'|$.
By choice of $k$, we have $|\taxa'| < \alpha(k+1)$.
Then we have

\begin{align*}
   \dmp(T_1,T_2)/\alpha & = \dmp(T_1',T_2')/\alpha \\
   & \leq |\taxa'|/\alpha \\
   & < \alpha(k+1)/\alpha = k+1 \\
   & \leq {\dmp}_{\cha}(T_1,T_2) + 1 \\ 
   & \leq (1+1/r){\dmp}_{\cha}(T_1,T_2) 
\end{align*}

Thus 
$\dmp(T_1,T_2)/(1+1/r)\alpha 
\leq {\dmp}_{\cha}(T_1,T_2)$, as required.
\end{proof}


\subsection{Bounding the distance between $\dtbr$ and $d_{MP}$
}
\label{sec:dtbr}

\emph{Tree Bisection and Reconnection} (TBR) distance, denoted
$\dtbr$, is a distance measure defined on two unrooted binary phylogenetic trees $T_1$, $T_2$. It is defined as the minimum number of ``TBR-moves'' required to transform $T_1$ into $T_2$ (or vice-versa): it is a metric \cite{AllenSteel2001}. Informally, a TBR-move consists of deleting an edge of a tree and then reconnecting  the two resulting components via a new edge. This definition is motivated by the way software for constructing phylogenetic trees heuristically navigates through tree space in search of better trees \cite{katherine2017shape}. However, for algorithmic and analytical purposes $d_{TBR}$ is most interesting because of its equivalence to the \emph{agreement forest} abstraction.
An agreement forest of $T_1$ and $T_2$ on the same set of taxa $X$ is a partition of $X$ into blocks $S_{1}, S_{2}\ldots,{S_t}$ such that: (1) for each $i$, $T_1|_{S_i} = T_2|_{S_i}$; (2)
$S_1, S_2, \dots, S_t$ are spanning-disjoint in $T_1$ and in $T_2$.
An \emph{(unrooted) maximum agreement forest} is an agreement forest with a minimum number of blocks, and $d_{TBR}(T_1,T_2)$ is equal to this minimum, minus 1 \cite{AllenSteel2001}. A maximum agreement forest for the two trees in Figure~\ref{fig:exampleMP} consists of three blocks $\{a,b\}$, $\{f,g\}$ and $\{c,d,e\}$, so here $d_{TBR}$ is 2.

The characterization of $d_{TBR}$ via agreement forests is significant, because agreement forests have opened the door to a large number of positive FPT and approximation results in the phylogenetics literature, and they have also attracted attention from outside phylogenetics. We refer to \cite{whidden2013fixed,downey2013fundamentals,van2016hybridization,chen2016approximating,bordewich2017fixed,shi2018parameterized,extremal2019} for recent results.
Moreover, a number of other problems have been shown to be FPT when parameterized by $d_{TBR}$, by leveraging properties of the $d_{TBR}$ kernel \cite{kelk2016reduction} and/or showing that, via agreement forests, the treewidth of a certain auxiliary graph structure is bounded by a function of $d_{TBR}$ (see the next section) \cite{kelk2015}. $d_{TBR}$ is a lower bound on many phylogenetic dissimilarity measurements \cite{kelk2015}, which helps to prove FPT results for these larger parameters, but what about $d_{MP}$?   It has previously been shown that $\dmp(T_1,T_2) \leq \dtbr(T_1,T_2)$ for any pair of trees $T_1, T_2$ \cite{fischerNonbinary,dMP-moulton}. However, the possibility remained that $d_{MP}$ could be arbitrarily smaller than $d_{TBR}$, and this hinders our ability to bind $d_{MP}$ to other phylogenetic parameters. Our contribution is to show that $\dmp$ and $\dtbr$ are in fact within a constant factor of each other: $\dtbr(T_1,T_2) \leq 2\alpha \dmp(T_1,T_2)$.

To show this, we use the fortunate fact that Reduction Rules~\ref{rule:cherry} and~\ref{rule:chain}, which we used to prove the kernel bound for \problemName, preserve $\dtbr$ as well as $\dmp$
for $\dtbr$. The following theorem is, modulo a small modification, due to \cite{AllenSteel2001}.

\begin{theorem}\label{thm:tbrReduction}
Let $(T_1',T_2')$ be a pair of phylogenetic trees on $\taxa'$ derived from $(T_1, T_2)$ by an application of Reduction Rule~\ref{rule:cherry} or~\ref{rule:chain}.
Then $\dtbr(T_1', T_2') = \dtbr(T_1, T_2)$.
\end{theorem}
\begin{proof}
Theorem 3.4 of \cite{AllenSteel2001} shows that $d_{TBR}$ is preserved under reduction rules similar to  Reduction Rules~\ref{rule:cherry} and~\ref{rule:chain}, except that common chains are reduced to length $3$ instead of $4$. For a pair of trees $T_1,T_2$ on $\taxa$, let $(T_1'', T_2'')$ with leaf set $\taxa''$ be the instance derived from $(T_1, T_2)$ by exhaustively applying these reduction rules. Also let $(T_1',T_2')$ with leaf set $\taxa'$ be the instance derived from $(T_1,T_2)$ by exhaustively applying  Reduction Rules~\ref{rule:cherry} and~\ref{rule:chain}. Observe that we may assume $\taxa'' \subseteq \taxa' \subseteq X$, since any leaf deleted in an application of   Reduction Rule~\ref{rule:cherry} or~\ref{rule:chain} can also be deleted by an application of one of the reduction rules in~\cite{AllenSteel2001}.
Furthermore by Lemma 2.1 of \cite{AllenSteel2001},  $d_{TBR}$ distance is non-increasing on subtrees induced by subsets of $\taxa$, which implies that $d_{TBR}(T_1'', T_2'') \leq  d_{TBR}(T_1', T_2') \leq d_{TBR}(T_1, T_2)$. As Theorem 3.4 of \cite{AllenSteel2001} states that $d_{TBR}(T_1'', T_2'') = d_{TBR}(T_1,T_2)$,  the chain of inequalities becomes a chain of equalities and hence $d_{TBR}(T_1', T_2') = d_{TBR}(T_1, T_2)$.

\end{proof}



\begin{theorem}\label{lem:tbrApprox}
For any pair of phylogenetic trees $T_1, T_2$ such that $T_1\neq T_2$, whence $\dmp(T_1,T_2)\geq 1$,
$$1 \leq \frac{\dtbr(T_1,T_2)}{\dmp(T_1,T_2)} \leq 2\alpha.$$
\end{theorem}
\begin{proof}
Let $(T_1', T_2')$ be the pair of trees derived from $(T_1,T_2)$ by exhaustively applying Reduction Rules~\ref{rule:cherry} and~\ref{rule:chain}, and let $\taxa'$ be the leaf set of $T_1'$ and $T_2'$.
It is well-known that $\dtbr(T_1',T_2') \leq |\taxa'|-3$ 
\cite{AllenSteel2001}.
Then by  Theorems~\ref{thm:kernelBound},~\ref{thm:dmpReduction} and~\ref{thm:tbrReduction}, 
\begin{align*}
\dtbr(T_1,T_2) & = \dtbr(T_1',T_2')  < |X'| \\ &  <  \alpha(\dmp(T_1',T_2')+1) \\ & \leq 2\alpha\ \dmp(T_1,T_2). 
\end{align*}
Using $\dmp(T_1,T_2) \leq \dtbr(T_1,T_2)$ \cite[Lemma 2.1]{dMP-moulton}, we have 
    $\dmp(T_1,T_2) \leq \dtbr(T_1,T_2) \leq 2\alpha\dmp(T_1,T_2)$, which, dividing by $\dmp(T_1,T_2)$, proves the theorem.
\end{proof}

\subsection{The treewidth of the display graph}
\label{sec:tw}

Let $G=(V,E)$ be an undirected graph. A \emph{tree decomposition} of $G$ consists of a multi-set of \emph{bags}, $B = \{B_1, \ldots, B_t\}$ where each
$B_i \subseteq V$, and a tree $T$ whose nodes
are in bijection with $B$, such that: (1) Every vertex $v \in V$ is in at least one bag; (2) for every edge $\{u,v\}$, at least one bag contains both $u$ and $v$, and (3) for every vertex $v \in V$, the bags of $T$ that contain $v$ induce a connected subtree of $T$. The \emph{width} of the tree decomposition is equal to the size of its largest bag, minus one, and the \emph{treewidth} of $G$ is the minimum width, ranging over all tree decompositions $T$ of $G$ \cite{Bodlaender96}. Treewidth derives its importance in combinatorial optimization from the fact that many NP-hard problems on graphs become fixed parameter tractable when parameterized by the treewidth of the graph \cite{bodlaender1994tourist}.

Given two phylogenetic trees $T_1, T_2$ on $X$, where $|X| \geq 3$, the \emph{display graph} of $T_1$ and $T_2$, denoted $D(T_1,T_2)$, is the graph obtained by identifying the leaves of $T_1$ and $T_2$ with the same label. A sequence of articles have studied the treewidth of display graphs, expressed as a function of various phylogenetic parameters, and used this to prove FPT results for a number of NP-hard phylogenetics problems using Courcelle's Theorem \cite{bryant2006compatibility,kelk2015,janssen2019treewidth} and explicit dynamic programming algorithms running over tree decompositions of the display graph \cite{baste2017efficient}. However, the question remained whether the treewidth of the display graph, denoted by $tw(D(T_1,T_2))$ could be bounded by a function of $d_{MP}(T_1,T_2)$ \cite{kelk2016reduction}.

The answer is emphatically yes: here we show, by leveraging the fact that $d_{MP}$ and $d_{TBR}$ are within a constant factor of each other, that the display graph has treewidth bounded by a linear function of $\dmp(T_1,T_2)$. 

\begin{theorem}
For two phylogenetic trees $T_1,T_2$ on $\taxa$,  $$tw(D(T_1,T_2)) \leq 2\alpha\dmp(T_1,T_2) + 2$$
\end{theorem}
\begin{proof}
It was shown in \cite{kelk2015} that $tw(D(T_1,T_2)) \leq \dtbr(T_1,T_2) + 2$.
As Theorem~\ref{lem:tbrApprox} shows $\dtbr(T_1,T_2) \leq 2\alpha\dmp(T_1,T_2)$ the theorem follows. 
\end{proof}

Note that Theorem 7.2 of \cite{kelk2017treewidth} shows an 
infinite family of trees where the 
treewidth of the display graph is 3
but $d_{MP}$ is unbounded.

\section{Conclusion}
\label{sec:concl}

A natural question is how far the analysis can be tightened, or changed, to improve the existing bound on the size of the kernel. In any case, it can be shown that for these two reduction rules a bound smaller than $20k-12$ is not possible. That is because the family of fully-reduced instances described in \cite{tightkernel} have exactly $15k-9$ taxa, where in this specific case $k=d_{TBR}=d_{MP}$. By replacing the length-3 chains with length-4 chains in this family we obtain the bound $20k-12$. 
We expect that, \emph{in practice}, the achieved reduction on realistic trees will be far superior to the bounds proven in this paper.

From the perspective of algorithm design it would be useful to design an explicit algorithm with FPT runtime that does not rely on kernelization; for example, by branching or by dynamic programming over an appropriately defined decomposition. Similarly, in the quest for small constant approximation factors it would be interesting to design polynomial-time  approximation algorithms that do not rely on kernelization. It is unlikely that through kernelization we will be able to achieve such truly small constant ratios.

The precise relationship between $d_{MP}$ and $d_{TBR}$ remains intriguing. Although we have now established that they are within a constant factor of each other, we are still a long way from proving or disproving the conjecture that $d_{MP} \geq (1/2)d_{TBR}$ \cite{kelk2016reduction}. An infinite family of examples is known where $d_{MP} = (1/2)d_{TBR} + o(1)$ \cite[Theorem 7.1]{dMP-moulton}, and small examples are known where $d_{MP} = (1/2)d_{TBR}$ (see e.g. Figure~\ref{fig:exampleMP}, based on \cite[Figure 5]{kelk2016reduction}), so $d_{MP} \geq (1/2)d_{TBR}$ would be the best possible bound.

On a slightly different note, recent publications have reduced the $d_{TBR}$ kernel size from $28k$ to $15k-9$ \cite{tightkernel}, and then to $11k-9$ \cite{kelk2019new}. The $11k-9$ kernel augments the two reduction rules discussed in this article, with five new reduction rules. Which of these new reduction rules work (possibly in a modified form) for $d_{MP}$, and how might this help us obtain a smaller linear kernel for $d_{MP}$? 

Finally, we note that there are several slight variations of $d_{MP}$ in the literature. These include the ``asymmetric" version $\damp(T_1, T_2) := \max_{\cha} (\PS{\cha}{T_1} - \PS{\cha}{T_2})$, in which $T_1$ is required to have the higher parsimony score,
and the ``restricted states" version $\dmp^2(T_1,T_2) := \max_{\cha}{\dmp}_{\cha}(T_1,T_2)$, where the maximum is taken over all characters with at most $2$ states~\cite{kelk2015,kelk2014complexity}.
Many of the results in this article will go through for $\damp(T_2,T_1)$, as the characters we construct consistently give a larger score to $T_2$. 
It is less obvious how our results impact on $\dmp^2$. In particular, it is not immediately clear whether the reduction rules described in \cite{kelk2016reduction} go through for $\dmp^2$, or how one would prove an analogue of Lemma~\ref{lem:wellBehavedSets} for $\dmp^2$.
Relatedly, it is unclear how much smaller 
$\dmp^2$
can be than $d_{MP}$ itself. Specifically, how important are additional states when attempting to maximize the parsimony distance between trees? It is known that $7d_{MP}-5$ states are sufficient to obtain a character that witnesses $d_{MP}$ \cite{boes2016linear}, but it is unclear what happens below this bound.


\section{Acknowledgements}
This work was supported by the Netherlands Organisation for Scientific Research (NWO) through Gravitation Programme Networks 024.002.003.



\bibliographystyle{plain}
\bibliography{bibliographyTOP}

\newpage

\appendix

\section{Finding optimal $d_1,d_2$}\label{sec:optimalds}
For the sake of completeness, we here argue that the choice of $d_1 = 4, d_2 = 5$ gives the minimum value of $\alpha = 2\cdot(9(d_1+d_2)-11)\cdot \lceil\frac{2d_1d_2+d_1}{d_1d_2-d_1-d_2} \rceil$ in Theorem~\ref{thm:kernelBound}.
Let $c = 9(d_1+d_2)-11$ and $t' = \lceil\frac{2d_1d_2+d_1}{d_1d_2-d_1-d_2} \rceil$, so that $\alpha = 2ct'$.
For $d_1 = 4, d_2 = 5$, we have $c = 81-11 = 70$ and $t' = \lceil \frac{44}{11}\rceil = 4$, and so $\alpha = 2 \cdot 70 \cdot 4 = 560$.
Figures~\ref{fig:c_vals},~\ref{fig:t_vals} and~\ref{fig:alpha_vals} gives the possible values of $c, t'$ and $\alpha$ respectively, for $d_1,d_2$ taking values between $2$ and $9$ (recall that $d_1,d_2$ must be at least $2$, as Lemma~\ref{lem:kernelBound} requires $d_1d_2 - d_1 - d_2 > 0$).

By inspection of Figure~\ref{fig:alpha_vals}, it is easy to see that the minimum possible value of $\alpha$ for $2 \leq d_1,d_2 \leq 9$ is $560$.
For larger values of $d_1,d_2$, we argue as follows: Observe that  $t' = \lceil\frac{2d_1d_2+d_1}{d_1d_2-d_1-d_2} \rceil$ is at least $3$ for any $d_1,d_2$, as $\frac{2d_1d_2+d_1}{d_1d_2-d_1-d_2} > \frac{2d_1d_2}{d_1d_2} = 2$. If one of $d_1,d_2$ is at least $10$, then $c = 9(d_1 + d_2) - 11 \geq 9(10+2) - 11 = 97$. But then for such values we would have $\alpha = 2ct' \geq 2\cdot97\cdot3 = 582$. Thus, the smallest value of $\alpha$ is in fact $560$, achieved for $d_1 = 4, d_2 = 5$.

\begin{figure}[h]
    \centering
\begin{tabular}{cc|c|c|c|c|c|c|c|c|l}
\cline{3-10}
& & \multicolumn{8}{ c| }{$d_2$} \\ \cline{3-10}
& & 2 & 3 & 4 & 5 & 6 & 7 & 8 & 9 \\ \cline{1-10}
\multicolumn{1}{ |c  }{\multirow{8}{*}{$d_1$} } &
\multicolumn{1}{ |c| }{2}  & -  & 34 & 43 & 52 & 61 & 70 & 79 & 88   \\ \cline{2-10}
\multicolumn{1}{ |c  }{}                        &
\multicolumn{1}{ |c| }{3}  & 34 & 43 & 52 & 61 & 70 & 79 & 88 & 97    \\ \cline{2-10}
\multicolumn{1}{ |c  }{}                        &
\multicolumn{1}{ |c| }{4}  & 43 & 52 & 61 & 70 & 79 & 88 & 97 & 106 \\ \cline{2-10}
\multicolumn{1}{ |c  }{}                        &
\multicolumn{1}{ |c| }{5} & 52 & 61 & 70 & 79 & 88 & 97 & 106 & 115  \\ \cline{2-10}
\multicolumn{1}{ |c  }{}                        &
\multicolumn{1}{ |c| }{6}  & 61 & 70 & 79 & 88 & 97 & 106 & 115 & 124  \\ \cline{2-10}
\multicolumn{1}{ |c  }{}                        &
\multicolumn{1}{ |c| }{7}  & 70 & 79 & 88 & 97 & 106 & 115 & 124 & 133  \\ \cline{2-10}
\multicolumn{1}{ |c  }{}                        &
\multicolumn{1}{ |c| }{8}  & 79 & 88 & 97 & 106 & 115 & 124 & 133 & 142  \\ \cline{2-10}
\multicolumn{1}{ |c  }{}                        &
\multicolumn{1}{ |c| }{9}  & 88 & 97 & 106 & 115 & 124 & 133 & 142 & 151 \\ \cline{1-10}
\end{tabular}
    \caption{Values for $c = 9(d_1+d_2)-11$}
    \label{fig:c_vals}
\end{figure}

\begin{figure}[h]
    \centering
\begin{tabular}{cc|c|c|c|c|c|c|c|c|l}
\cline{3-10}
& & \multicolumn{8}{ c| }{$d_2$} \\ \cline{3-10}
& & 2 & 3 & 4 & 5 & 6 & 7 & 8 & 9 \\ \cline{1-10}
\multicolumn{1}{ |c  }{\multirow{8}{*}{$d_1$} } &
\multicolumn{1}{ |c| }{2}  & -  & 14 & 9 & 8 & 7 & 6 & 6 & 6   \\ \cline{2-10}
\multicolumn{1}{ |c  }{}                        &
\multicolumn{1}{ |c| }{3}  & 15 & 7 & 6 & 5 & 5 & 5 & 4 & 4    \\ \cline{2-10}
\multicolumn{1}{ |c  }{}                        &
\multicolumn{1}{ |c| }{4} & 10 & 6 & 5 & 4 & 4 & 4 & 4 & 4 \\ \cline{2-10}
\multicolumn{1}{ |c  }{}                        &
\multicolumn{1}{ |c| }{5}  & 9 & 5 & 5 & 4 & 4 & 4 & 4 & 4  \\ \cline{2-10}
\multicolumn{1}{ |c  }{}                        &
\multicolumn{1}{ |c| }{6}  & 8 & 5 & 4 & 4 & 4 & 4 & 3 & 3  \\ \cline{2-10}
\multicolumn{1}{ |c  }{}                        &
\multicolumn{1}{ |c| }{7}  & 7 & 5 & 4 & 4 & 4 & 3 & 3 & 3  \\ \cline{2-10}
\multicolumn{1}{ |c  }{}                        &
\multicolumn{1}{ |c| }{8}  & 7 & 5 & 4 & 4 & 4 & 3 & 3 & 3  \\ \cline{2-10}
\multicolumn{1}{ |c  }{}                        &
\multicolumn{1}{ |c| }{9}  & 7 & 5 & 4 & 4 & 3 & 3 & 3 & 3  \\ \cline{2-10} \cline{1-10}
\end{tabular}
    \caption{Values for $t' = \lceil\frac{2d_1d_2+d_1}{d_1d_2-d_1-d_2} \rceil$}
    \label{fig:t_vals}
\end{figure}


\newpage    

\begin{figure}[h]
    \centering
\begin{tabular}{cc|c|c|c|c|c|c|c|c|l}
\cline{3-10}
& & \multicolumn{8}{ c| }{$d_2$} \\ \cline{3-10}
& & 2 & 3 & 4 & 5 & 6 & 7 & 8 & 9 \\ \cline{1-10}
\multicolumn{1}{ |c  }{\multirow{8}{*}{$d_1$} } &
\multicolumn{1}{ |c| }{2}  & - & 952 & 774 & 832 & 854 & 840 & 948 & 1056   \\ \cline{2-10}
\multicolumn{1}{ |c  }{}                        &
\multicolumn{1}{ |c| }{3}  & 1020 & 602 & 624 & 610 & 700 & 790 & 704 & 776     \\ \cline{2-10}
\multicolumn{1}{ |c  }{}                        &
\multicolumn{1}{ |c| }{4} & 860 & 624 & 610 & {\bf 560} & 632 & 704 & 776 & 848\\ \cline{2-10}
\multicolumn{1}{ |c  }{}                        &
\multicolumn{1}{ |c| }{5} & 936 & 610 & 700 & 632 & 704 & 776 & 848 & 920 \\ \cline{2-10}
\multicolumn{1}{ |c  }{}                        &
\multicolumn{1}{ |c| }{6}  & 976 & 700 & 632 & 704 & 776 & 848 & 690 & 744  \\ \cline{2-10}
\multicolumn{1}{ |c  }{}                        &
\multicolumn{1}{ |c| }{7}  & 980 & 790 & 704 & 776 & 848 & 690 & 744 & 798  \\ \cline{2-10}
\multicolumn{1}{ |c  }{}                        &
\multicolumn{1}{ |c| }{8} & 1106 & 880 & 776 & 848 & 920 & 744 & 798 & 852  \\ \cline{2-10}
\multicolumn{1}{ |c  }{}                        &
\multicolumn{1}{ |c| }{9} & 1232 & 970 & 848 & 920 & 744 & 798 & 852 & 906  \\ \cline{2-10}\cline{1-10}
\end{tabular}
    \caption{Values for $\alpha = 2\cdot(9(d_1+d_2)-11)\cdot \lceil\frac{2d_1d_2+d_1}{d_1d_2-d_1-d_2} \rceil$. Observe that the minimum is acheived at $d_1 = 4, d_2 = 5$.}
    \label{fig:alpha_vals}
\end{figure}

\end{document}